\documentclass[11pt]{article}

\usepackage[T1]{fontenc}
\usepackage[utf8]{inputenc}
\usepackage{lmodern}
\usepackage{microtype}

\usepackage[margin=1in]{geometry}
\usepackage{setspace}
\usepackage{amsmath}
\usepackage{amssymb}
\usepackage{amsthm}
\usepackage{algorithm}
\usepackage{algpseudocode}
\usepackage{listings}
\lstdefinestyle{skill}{%
  basicstyle=\footnotesize\ttfamily,
  breaklines=true, breakatwhitespace=false, columns=fullflexible,
  keepspaces=true, showstringspaces=false,
  frame=single, framesep=5pt, xleftmargin=8pt, xrightmargin=4pt,
  aboveskip=8pt, belowskip=4pt,
  literate={—}{{\textemdash}}1 {κ}{{\ensuremath{\kappa}}}1 {→}{{\ensuremath{\rightarrow}}}1 {…}{{\ldots}}1
}

\usepackage{booktabs}
\usepackage{array}
\usepackage{graphicx}
\graphicspath{{figures/}}
\usepackage{caption}
\usepackage{xcolor}

\usepackage{enumitem}
\usepackage[round]{natbib}
\definecolor{linkblue}{RGB}{0,68,170}
\usepackage[colorlinks=true,citecolor=linkblue,urlcolor=linkblue,linkcolor=linkblue,
            breaklinks=true]{hyperref}

\theoremstyle{definition}
\newtheorem{definition}{Definition}
\theoremstyle{plain}
\newtheorem{proposition}{Proposition}

\newcommand{\passat}[1]{\textup{pass}@#1}
\newcommand{\arm}[2]{\mathrm{#1}@#2}
\newcommand{\Fjudge}{F_{\textup{judge}}}
\newcommand{\Frand}{F_{\textup{rand}}}
\newcommand{\E}{\mathbb{E}}
\newcommand{\Prob}{\mathbb{P}}

\title{\bfseries Scaffold, Not Vocabulary? A Controlled, Two-Tier,
Pre-Registered Study of a Popperian Code-Generation Skill}

\author{Mehmet \.{I}\c{s}can\thanks{Corresponding author. PythaLab, Y\i{}ld\i{}z
Technical University, Istanbul, Turkey. Email: \href{mailto:miscan@yildiz.edu.tr}{\texttt{miscan@yildiz.edu.tr}}.}\\[2pt]
\small PythaLab, Y\i{}ld\i{}z Technical University, Istanbul, Turkey}
\date{}

\begin{document}
\maketitle

\begin{abstract}
\noindent
Large language models (LLMs) increasingly write, review, and judge code, and a
fast-growing practice equips them with prompt ``skills'': system-prompt modules that ask
the model to plan, self-critique, or reason like a scientist. A prominent example tells
the model to behave as a \emph{Popperian falsificationist} that attacks its own
conjectures with the most severe tests it can devise, and such skills are reported to
improve generated code. The reported improvement, however, is almost always read off an
LLM-as-a-judge---an instrument with documented positional, self-preference, and stylistic
biases whose verdicts routinely collapse once a real control is added. We therefore ask a
sharper question than ``does the skill help?'': \emph{if} it appears to help, is the gain
produced by the skill's Popperian \emph{procedural content} or merely by the \emph{structure}
that any reasonable scaffold imposes? To separate the two we pre-register a two-tier ablation with
three controls: a length-matched \textbf{placebo}, a \textbf{labels-only} scaffold that
keeps the Popperian headers but strips the procedure, and an \textbf{execution oracle}
(HumanEval+ unit tests), augmented by a vocabulary-halo sentinel and a same-model
self-judge audit, and we deliberately anchor the design on a very small model, where a
genuine effect has the most room to show itself. On a frontier model (Claude Sonnet~4.6,
$N{=}163$) all four conditions sit near the benchmark ceiling ($\textup{V}{=}95.1\%$,
$\textup{F}{=}95.7\%$, $\textup{L}{=}95.7\%$, $\textup{P}{=}96.9\%$) and do not separate
(all six pairwise contrasts within $\pm2$ points, none statistically significant under paired
McNemar tests), so the pre-registered $+5$-point improvement hypothesis is not supported---though
we stress this is a \emph{ceiling-limited} non-detection, because a baseline already at $95.1\%$
leaves almost no room for a $5$-point gain. On the small model (Qwen2.5-Coder-0.5B, $N{=}164$;
$5{,}248$ generations and $164$ self-judges), where the baseline is far from any ceiling,
structured arms lift best-of-eight correctness by roughly $20$--$22$ points over $\arm{V}{8}$
(F and L by $+22.0$, the placebo by $+19.5$), but the full Popperian
skill shows no separable benefit over a labels-only scaffold: their aggregate best-of-eight
pass rates coincide ($\arm{F}{8}{=}\arm{L}{8}{=}56.7\%$ vs.\ $\arm{V}{8}{=}34.8\%$; though
the two disagree on $24$ individual problems) and the length-matched placebo trails by only
$2.4$ points (paired $p{=}0.60$). A same-model $0.5$B self-judge applying the Popperian rubric
does not beat random selection among its own samples ($25.6\%$ vs.\ a $24.9$--$26.8\%$ random
baseline) and concentrates $60\%$ of its picks on a single candidate index---a pattern
consistent with position bias, though we did not run the order-randomized replication that
would identify it cleanly. We read these results narrowly: in the two settings tested, the
full skill's Popperian \emph{procedural content} adds no separable execution-correctness
benefit beyond a labels-only scaffold, so the measured gains are attributable to scaffold
\emph{structure} rather than to that content. We contribute a calibrated negative result and
a reusable disambiguation protocol (placebo, labels-only, execution oracle, halo sentinel);
the finding bounds a specific engineering claim about one prompt-skill family and is not an
evaluation of Popperian methodology, whose meta-level use was load-bearing in our own design.

\end{abstract}

\section{Introduction}
\label{sec:intro}

Large language models have moved, in only a few years, from autocompleting single lines
to drafting, reviewing, and merging whole software changes. As they were deployed, a folk
engineering discipline grew up around them: practitioners discovered that \emph{how} you
ask matters as much as \emph{what} you ask, and they began packaging reusable
instructions (``prompt skills'') that tell a model to plan before coding, to enumerate
edge cases, to criticize its own draft, or to reason in the style of some expert. These
skills are cheap to write, easy to share, and routinely reported to lift quality. The most
intellectually attractive of them borrow the language of science itself: they ask the
model to act as a \emph{Popperian falsificationist}, to treat each candidate program as a
conjecture and to attack it with the most severe tests it can imagine before declaring
success \citep{popper1959logic,popper1963conjectures}. The appeal is clear: a model that
internalized the scientist's habit of seeking refutation rather than confirmation might
catch its own bugs.

The trouble is in the measurement. A skill's reported benefit is almost always established
by an \emph{LLM-as-a-judge}: a second model scores the output against a rubric, and the
skill ``wins'' if its scores are higher. But the judge is a fragile instrument. At the
coarsest level it is sensitive to the order in which candidates are presented:
\citet{wang2024faireval} find that simply swapping two answers flips the preferred one in
$82.5\%$ of comparisons for one judge and $46.3\%$ for a stronger one. Judges prefer
outputs from their own model family \citep{yang2026selfpref,spiliopoulou2025playfav} and (most relevant here) they reward surface features such as fluency, length, and
recognizable jargon over substance: \citet{wu-aji-2025-style} show that a longer answer
carrying minor factual errors is scored \emph{above} a shorter correct one (Elo $1206$
vs.\ $1096$). The pathology reaches into code specifically.
\citet{moon2026judgecover} report that code judges shift their verdicts by up to $26.7$
percentage points under content-preserving, cosmetic perturbations, and
\citet{zhao2026biasloop} report that injecting authority- or reasoning-flavored cues into a
judge prompt drags accuracy on identical code from $77.5\%$ down to $61.0\%$. A judge that
can be moved this far by words alone is the wrong instrument for deciding whether
\emph{words}, a skill's vocabulary, improved a program.

Zoom out from judges to the broader prompting literature and the picture sharpens into a
warning. A growing body of careful re-evaluations finds that celebrated prompt-engineering
effects shrink or vanish once a proper control is introduced.
\citet{vaugrante2024replication} re-run five popular techniques across six models and five
benchmarks and find that most, including chain-of-thought, produce no reliable gain
(chain-of-thought: $\approx\!0\%$ on average, $p{=}.8$), with only isolated,
architecture-specific exceptions. \citet{wang2025advancedpe} show that on high-capability
models the value of elaborate prompting collapses, and what survives is not phrasing but
\emph{execution feedback}. The constructive counterpart to these null results is an
emerging consensus about \emph{where} prompting effects actually come from: structure, not
wording. \citet{sclar2024formspread} demonstrate that meaning-preserving formatting choices
alone (spacing, separators, casing) move accuracy by as much as $76$ points;
\citet{fagadau2024prompt}, in a factorial study of $124{,}800$ Copilot prompts, find that
\emph{structural} features (worked examples, a task summary) significantly improve
correctness ($p{=}10^{-4}$ for examples) while \emph{stylistic} features (mood, tense) do
not; and \citet{khojah2025impact} confirm that adding a persona or a chain-of-thought
flourish on top of a structural scaffold yields no statistically significant correctness
gain. \citet{akli2026underspecification} sharpen the boundary: vocabulary perturbations
bite only when a prompt is structurally thin (a $-7.1\%$ average drop on the lightly
specified HumanEval, but only $-0.2\%$ on the richly specified LiveCodeBench); once the
scaffold is rich, the words barely matter.

Read together, this literature poses a clean, unanswered question for any
methodology-flavored skill. Suppose a Popperian coding skill appears to help. The apparent
help could be a judge artifact (the rubric rewards the words ``falsifiability'' and
``severity''); it could be a structural effect (the skill, like any scaffold, makes the
model enumerate steps and cases); or it could be the genuine article (the falsification
\emph{content} causes better code). Disentangling these requires three controls that prior
code-generation studies almost never deploy together: a \textbf{length-matched placebo}
that preserves prompt length and best-practice tone while removing all Popperian content; a
\textbf{labels-only} scaffold that keeps the Popperian section headers but strips the
procedure beneath them, so that the full skill can be compared against a fixed labelled
scaffold rather than against vanilla alone; and an
\textbf{execution oracle} that scores functional correctness with unit tests rather than
rubric impressions. The components of this design exist separately in the
literature (surface-form judge bias \citep{moon2026judgecover}, self-preference in
same-model judges \citep{yang2026selfpref,spiliopoulou2025playfav}, structure-versus-style
ablations \citep{fagadau2024prompt,khojah2025impact,akli2026underspecification}, and
formatting sensitivity \citep{sclar2024formspread}) but we are not aware of a prior
pre-registered study that combines a length-matched placebo, a labels-only scaffold, an
execution oracle, a vocabulary-halo sentinel, and a same-model self-judge audit to ask
whether a Popperian falsification skill adds value beyond plain structure. That integration
is the gap we fill.

We approached the problem philosophically as well as empirically, because the object of
study is itself a philosophical claim. Popper's program holds that a hypothesis earns
scientific standing not by accumulating confirmations but by surviving sincere attempts at
refutation; the \emph{degree of corroboration} $C(h,e)$ is explicitly not a probability of
truth \citep{popper1959logic}, and progress is framed as increasing \emph{verisimilitude}
\citep{popper1963conjectures,niiniluoto2014progress}. The error-statistical school makes the
operative idea precise: data warrant a hypothesis only if the test had a high probability of
exposing the hypothesis as false had it been false---Mayo and Spanos's \emph{severity
requirement} \citep{mayospanos2006severe,mayo2025severe}. A coding skill that merely
\emph{names} severity, falsifiability, and predictive risk has done none of this work;
whether the naming nevertheless changes behavior is an empirical question. There is an
irreducible reflexivity here that we embrace rather than hide: we evaluate a Popperian
\emph{skill} using a Popperian \emph{method}---pre-registered hypotheses, executable
counterexamples, and blinded outcome prediction. The skill is the object of measurement; the
method is the instrument, and our results show the two diverge.

A second design choice deserves its own justification, because it inverts the usual reflex
to test on the strongest available model. We test not only a frontier model but, centrally,
a very small one (a 0.5B-parameter local model). A weak model sits far below the benchmark
ceiling, which makes it a \emph{low-baseline} setting with substantial headroom to look for a
real effect: a prompt component that genuinely carries signal has ample room to produce a
large, unmistakable gain there. Our own data confirm the sensitivity of this testbed: the labels-only and
structured scaffold conditions lift best-of-eight performance by roughly $20$--$22$
points, so if the incremental Popperian content carried a correctness
signal, this is a setting where it would have ample room to show itself. Its absence at low
capability (the full skill does not separate from the labels-only scaffold) is therefore
more informative about a separable content effect than a high-capability ceiling result
alone could be. The same logic supplies a constructive screening principle that we make
explicit \emph{as a screen, not a proof of transfer}: a sensitive small-model testbed is a
reasonable place to \emph{screen} a candidate prompt component before committing
high-capability compute, and a component that does not move a low-baseline weak model with
ample headroom is, at the least, not demonstrated to be an active ingredient. We are careful about the limits of this principle (transfer is not guaranteed in either direction
\citep{kojima2022zeroshot,zhou2024lessreliable}, and capability-dependent effects in other
models or tasks are not ruled out) but as a negative screen it is useful, and it organizes
our two-tier design.

This reframing turns a vague ``does it work?'' into a falsifiable program. We pre-register
an ablation ladder---vanilla (V), labels-only (L), labels$+$discipline (LD),
labels$+$discipline$+$severity (LDS), and the full skill (F)---together with a
length-matched placebo (P) and a verificationist anti-skill (ANTI). We score a
high-capability model and a low-capability model against the same execution oracle, add a
reliability stage with a vocabulary-only \emph{halo sentinel}, and, at the low tier, test
whether the model can use its own Popperian rubric to select the best of several samples---a
miniature of test-time reward selection. Every hypothesis and threshold was committed to
public version control before the corresponding data were collected, following recent calls
for registered protocols in machine learning and NLP
\citep{vaccaro2026preregistration,sogaard2023twosided}.

\paragraph{Contributions and findings.}
This paper makes the following contributions; we state up front that the headline result is
a \emph{calibrated negative one}, and we are explicit about the conditions under which the
method fails.
\begin{enumerate}[leftmargin=1.4em,itemsep=2pt]
\item \textbf{A content-versus-structure ablation design} for methodology-flavored coding
skills: an additive ladder (V\,$\to$\,L\,$\to$\,LD\,$\to$\,LDS\,$\to$\,F) plus a
length-matched placebo and a verificationist anti-skill, with the $\textup{F}-\textup{L}$
contrast estimating the effect of the full Popperian procedural content added on top of a
fixed labels-only scaffold---a content-beyond-labels contrast, not a pure-vocabulary one
(\S\ref{sec:method}).
\item \textbf{A reliability stage with a vocabulary-halo sentinel} that detects a confound
intra-rater reliability misses entirely: our validated judge ($\textup{ICC}{=}0.959$) still
awards $10/20$ substance points to an output that contains the Popperian \emph{labels} and
nothing else.
\item \textbf{A two-tier correctness result.} At high capability the full skill shows no
separable correctness benefit over vanilla, the placebo, or labels-only (all within $2$
points, $N{=}163$); the pre-registered $+5$-point improvement hypothesis is not supported,
though this is a ceiling-limited non-detection rather than a demonstration of zero effect.
\item \textbf{A low-capability dissection,} positioned as a sensitive screen. Every structured arm
we tested lifts a $0.5$B model's best-of-eight correctness by roughly $20$--$22$ points, but the
full skill shows no separable benefit over labels-only (aggregate $\arm{F}{8}{=}\arm{L}{8}$) and
only a small, non-significant edge over a placebo---so the lift is carried by scaffold structure
rather than by the incremental Popperian content, in this setting.
\item \textbf{A small-scale self-judge audit:} a $0.5$B model applying its own Popperian
rubric to select among samples does not beat a random draw and concentrates its picks on one
candidate index, a result consistent with a verifier--generator capability gap and with
position bias (an order-randomized replication, which we did not run, would be needed to
identify the latter cleanly).
\item \textbf{A reusable disambiguation protocol} (placebo $+$ labels-only $+$ oracle $+$
halo sentinel) and a public release of pre-registrations, raw outputs, and analysis code,
so that future skill-effectiveness claims can be stated---and tested---on firmer ground.
\end{enumerate}
\noindent The boundary of the claim matters. We do \emph{not} conclude that Popperian
reasoning is worthless for code; we conclude that, \emph{in the two settings we tested}, the
full skill's procedural content showed no separable correctness benefit once structure was
controlled and correctness was measured by execution. Where the skill helped (the weak
model, under best-of-eight evaluation), a labels-only scaffold matched the aggregate lift, and the
corresponding single-sample labels/placebo comparison was not collected; where it might have helped
(the strong model), the benchmark ceiling left little room to detect it.

\section{Method}
\label{sec:method}

\subsection{The proposed framework}
\label{sec:framework}

Our object of study is a single Markdown ``skill'' loaded as the system prompt of a coding
agent. The skill instructs the agent to externalize a hypothesis about the required program,
to design severe tests against that hypothesis, to quantify ``severity'' and ``predictive
risk,'' and to revise before answering. The engineering claim attached to such skills is
that this Popperian \emph{content} improves the code. The framework we propose is built to
decide whether that claim survives once two well-known confounds (judge bias and prompt
structure) are removed, and it does so by combining an ablation ladder, two
confound-removing controls, two capability tiers, and an execution oracle into a single
pre-registered measurement instrument, sketched in Figure~\ref{fig:framework}.

\begin{figure}[t]
\centering
\includegraphics[width=\linewidth]{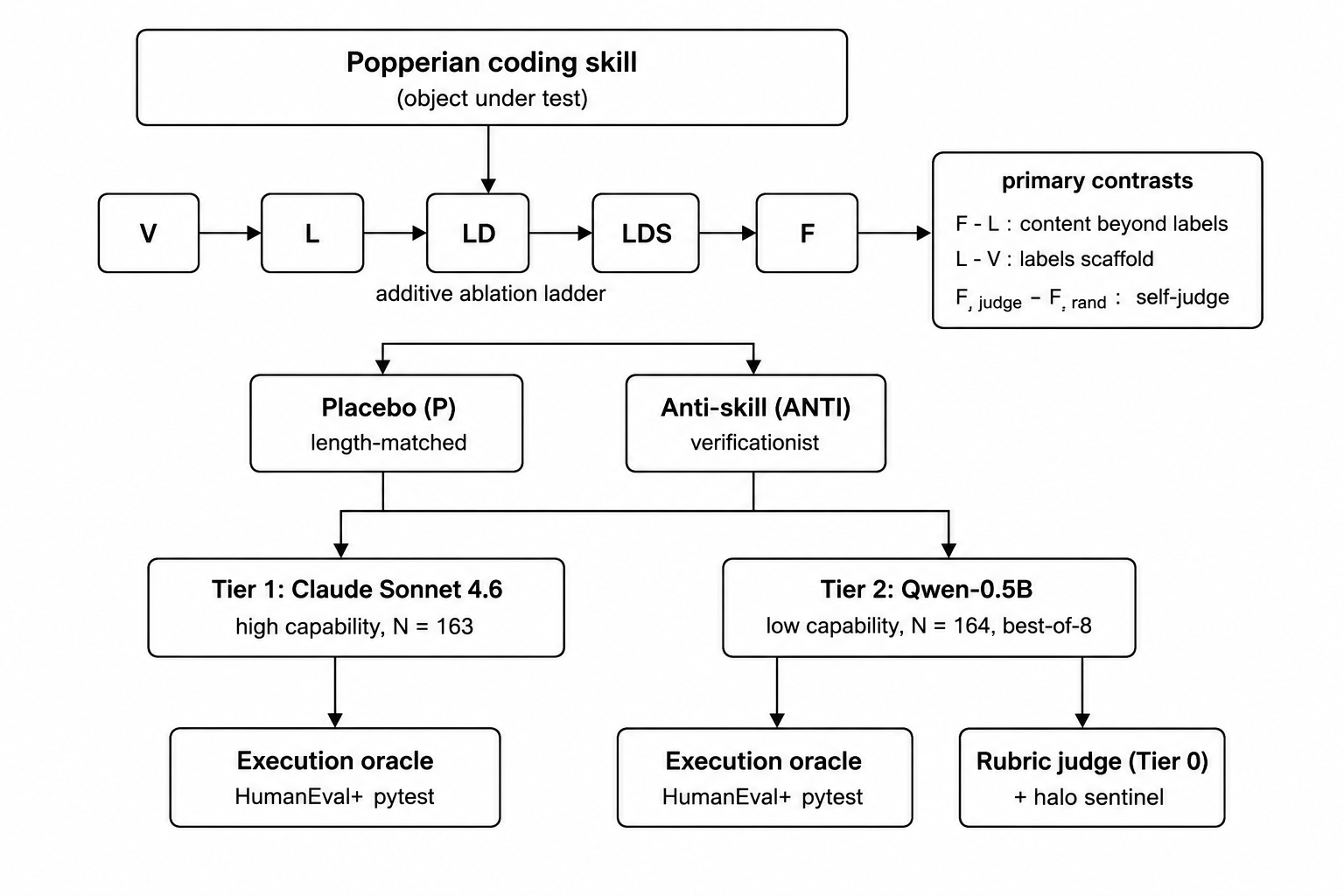}
\caption{The proposed framework. The Popperian skill is decomposed into an additive ablation
ladder (V\,$\to$\,L\,$\to$\,LD\,$\to$\,LDS\,$\to$\,F); a length-matched placebo (P) and a
verificationist anti-skill (ANTI) sit beside it as controls. Two capability tiers are each
scored by an execution oracle (correctness), with a separate rubric judge and
vocabulary-halo sentinel validated in Tier~0. The primary contrasts are
$\textup{F}-\textup{L}$ (full Popperian content added on top of a fixed labels-only scaffold;
not a pure-vocabulary contrast, see Proposition~\ref{prop:ident}), $\textup{L}-\textup{V}$
(labels/header scaffold over vanilla), and $\Fjudge-\Frand$ (whether a self-judge beats random selection).}
\label{fig:framework}
\end{figure}

The first part of the framework decomposes the skill into the seven nested conditions of
Table~\ref{tab:variants}. Each is a Markdown skill loaded as the system prompt of a coding
subagent; the ladder is \emph{additive}, so that consecutive conditions differ by exactly
one component and the marginal effect of that component is the difference of their outcomes.
Condition \textbf{L} keeps Popperian terminology in section headers only, with no procedure
beneath them; \textbf{LD} adds a numbered procedural scaffold; \textbf{LDS} adds a
severity-quantification formula; and \textbf{F} is the full skill. \textbf{P} is a
length-matched best-practices placebo that contains no Popperian content, and \textbf{ANTI}
is a verificationist inverse that instructs the model to seek confirmations. The single most
important comparison in the paper is \textbf{F versus L}: because L and F share the same
labels-only scaffold (the same headers, the same nominal sections) and differ in the
\emph{content} placed beneath those headers (the full falsification procedure, severity
formula, and predictive-risk narrative in F; nothing in L), their difference estimates the
marginal effect of that Popperian content on top of the scaffold that carries it. As
Proposition~\ref{prop:ident} makes explicit, this is a content-beyond-labels contrast rather
than a pure-vocabulary one (F and L also differ in procedure, instruction density, and
length) while the length-matched placebo P guards separately against a prompt-length
confound.

The seven conditions enter the two measurement stages asymmetrically, which we flag here
rather than leave implicit. The full ladder (including the intermediate rungs LD and LDS and
the anti-skill ANTI) is scored only at the rubric/anatomy stage (Tier~1.A1), where its purpose
is to localize \emph{which component} of the skill moves a rubric judge; the execution-oracle
stage (Tier~1.A2 and Tier~2) reports the four conditions---V, L, F, and P---for which complete
generation artifacts were collected and graded under the oracle. We make no execution-correctness
claim for LD, LDS, or ANTI, because they were not run through the oracle at scale; this is a
deliberate scope limit, restated in \S\ref{sec:discussion}.

\begin{table}[t]
\centering\small
\begin{tabular}{@{}llp{7.4cm}@{}}
\toprule
\textbf{Cond.} & \textbf{Adds} & \textbf{Description} \\
\midrule
V    & ---            & Vanilla; bare task, no skill \\
L    & labels         & Popperian terminology in headers only; no procedure \\
LD   & $+$ discipline & L $+$ a numbered procedural scaffold (no formulas) \\
LDS  & $+$ severity   & LD $+$ a severity-quantification formula \\
F    & $+$ meta       & Full Popperian skill (procedure, severity, predictive risk) \\
P    & (control)      & Length-matched best-practices placebo; no Popperian content \\
ANTI & (control)      & Verificationist anti-skill; seek confirmations \\
\bottomrule
\end{tabular}
\caption{Pre-registered conditions. The ladder V$\to$L$\to$LD$\to$LDS$\to$F is additive; the
\textbf{F\,vs.\,L} contrast estimates the full Popperian content added on top of a fixed
labels-only scaffold (a content-beyond-labels contrast, not a pure-vocabulary one;
Proposition~\ref{prop:ident}), \textbf{L\,vs.\,V} measures the labels/header scaffold effect over vanilla, and
\textbf{F\,vs.\,P} guards against a length confound.}
\label{tab:variants}
\end{table}

The second part of the framework states, precisely, what each contrast estimates. Let
$c \in \mathcal{C}=\{\textup{V},\textup{L},\textup{LD},\textup{LDS},\textup{F},\textup{P},
\textup{ANTI}\}$ index a condition, realized as a system prompt $s_c$. For a problem $q$
drawn from a benchmark $\mathcal{Q}$ and a generator $M$, let $Y(c,q)\in\{0,1\}$ be the
event that a sample produced under $\langle s_c, q\rangle$ passes the execution oracle. The
correctness of a condition is
\begin{equation}
\label{eq:pi}
\pi(c)=\E_{q\sim\mathcal{Q}}\,\E_{M}\,[\,Y(c,q)\,].
\end{equation}
Write each prompt as $s_c=(\sigma_c,\kappa_c)$, where $\sigma_c$ is a \emph{labels-only
scaffold} (the Popperian section headers and the output contract) and $\kappa_c$ is the
\emph{procedural content} placed beneath those headers (the numbered procedure, the
severity formula, the predictive-risk narrative, and the revision discipline). The ladder
holds the labels-only scaffold fixed across L and F ($\sigma_{\textup{L}}=\sigma_{\textup{F}}$)
and adds content in steps, while V is the empty scaffold; this yields a \emph{structure}
estimand and a \emph{content-beyond-labels} estimand,
\begin{align}
\Delta_{\mathrm{struct}}  &= \pi(\textup{L})-\pi(\textup{V}), &&\text{(labels-only scaffold over vanilla)} \label{eq:struct}\\
\Delta_{\mathrm{content}} &= \pi(\textup{F})-\pi(\textup{L}). &&\text{(full Popperian content over labels-only)} \label{eq:content}
\end{align}
We are deliberately conservative about how $\Delta_{\mathrm{content}}$ is read.

\begin{proposition}[Conservative interpretation of the F$-$L contrast]
\label{prop:ident}
$\Delta_{\mathrm{content}}=\pi(\textup{F})-\pi(\textup{L})$ in \eqref{eq:content} estimates the
incremental effect of adding the full Popperian procedural content on top of a fixed
labels-only scaffold. It is \emph{not} a pure vocabulary (terminology) effect: F and L differ
not only in Popperian wording but also in procedure, the severity formula, the predictive-risk
narrative, revision discipline, instruction density, and token length. A pure-vocabulary claim
would require further ablations that hold procedure, length, and behavioral instructions fixed
while varying only terminology.
\end{proposition}
\begin{proof}[Justification]
Because $\sigma_{\textup{L}}=\sigma_{\textup{F}}$, the contrast removes the scaffold-presence
confound that separates any skill from vanilla, so $\Delta_{\mathrm{content}}$ is not inflated
by ``has a scaffold versus has none.'' But the residual difference between $\kappa_{\textup{F}}$
and $\kappa_{\textup{L}}=\varnothing$ bundles several co-varying components, so the contrast is
\emph{conservative} for any single one of them---vocabulary included---and cannot attribute an
effect to terminology alone. The length-matched placebo supplies a second control:
$\pi(\textup{F})-\pi(\textup{P})$ holds prompt length and generic best-practice tone fixed while
removing Popperian content, so a credible content effect should survive \emph{both}
$\textup{F}-\textup{L}$ and $\textup{F}-\textup{P}$. The ablations needed to isolate terminology
specifically are enumerated in \S\ref{sec:discussion}. \end{proof}

Correctness under sampling is reported with the unbiased $\passat{k}$ estimator of
\citet{chen2021humaneval}: generating $n\ge k$ samples per problem and counting $c_q$
correct,
\begin{equation}
\label{eq:passk}
\passat{k}=\E_{q\sim\mathcal{Q}}\!\left[\,1-\binom{n-c_q}{k}\big/\binom{n}{k}\,\right],
\end{equation}
so that $\passat{1}$ is the single-sample (greedy) rate and $\passat{n}$ is best-of-$n$.
To make the philosophical target operational we adopt the error-statistical definition of
severity, which the skill's vocabulary gestures at but does not implement.

\begin{definition}[Operational severity, after \citealp{mayospanos2006severe}]
\label{def:severity}
A hypothesis $H$ (here ``this program is correct'') passes a test $T$ with data $x_0$
\emph{severely} when (S1) $x_0$ agrees with $H$, and (S2) with high probability $T$ would
have produced a result agreeing \emph{less} well with $H$ had $H$ been false:
\begin{equation}
\label{eq:sev}
\Prob\!\big(\,d(X,H) > d(x_0,H)\ \big|\ \neg H\,\big)\ \ge\ 1-\alpha,
\end{equation}
where $d(\cdot,H)$ measures disagreement with $H$.
\end{definition}

\noindent An execution oracle on a dense suite of adversarial unit tests \emph{approximates}
an operational severe-test regime in the sense of \eqref{eq:sev}: it raises the probability
that an incorrect program is exposed, because a wrong program is likely to fail at least one
augmented test. We are careful not to overstate the correspondence---the augmented tests are
an engineering analogue of severity, not a full instantiation of Mayo--Spanos severity, which
additionally requires a specified test distribution, an alternative error model, and an
explicit probability of detecting a discrepancy. A skill that merely \emph{writes} the word
``severity'' supplies none of this, which is why we measure correctness by execution and treat
rubric scores as suspect. That suspicion is itself quantified through a halo ratio.

\begin{definition}[Vocabulary-halo ratio]
\label{def:halo}
For a rubric judge $J$ and condition $c$, let $\bar J(c)$ be the mean substance-dimension
score of genuine outputs and $J(\mathrm{sent}_c)$ the score of a \emph{sentinel}---a
synthetic output containing the labels of $c$ but no code or procedure. The halo ratio is
\begin{equation}
\label{eq:halo}
\rho_c=\frac{J(\mathrm{sent}_c)}{\bar J(c)},
\end{equation}
and a pre-registered halo is declared when $\rho_c>\tau$ with $\tau=0.70$.
\end{definition}

Finally, ``no meaningful difference'' is operationalized by an equivalence rule rather than a
failure to reject. With a pre-registered margin $\delta=5$ points, a directional hypothesis
of the form $\pi(\textup{F})-\pi(\textup{V})\ge\delta$ is declared \emph{falsified} when the
observed difference is below $\delta$ and the upper confidence limit excludes it, and two
conditions are declared equivalent when the two one-sided tests (TOST) for paired
proportions reject at the $\pm\delta$ boundary \citep{tango1998equivalence}:
\begin{equation}
\label{eq:tost}
\text{equivalence}(c,b)\ \Longleftrightarrow\
\big[\widehat{\pi}(c)-\widehat{\pi}(b)\big]_{\text{CI}_{1-2\alpha}} \subset (-\delta,+\delta).
\end{equation}
Paired pass/fail outcomes use McNemar's test \citep{fagerland2013mcnemar}, and all intervals
are bootstrap $95\%$ confidence intervals \citep{diciccio1996bootstrap}.

The rubric judge that produces the Tier-0 evidence is itself validated before any reliance on
it. Tier-1.A1 rubric scores come from a Claude Sonnet~4.6 subagent applying a ten-dimension
rubric to each output, blinded to the condition. We measure intra-rater reproducibility by
the intraclass correlation coefficient under a two-way random-effects, absolute-agreement, single-measures model (ICC(2,1) in the nomenclature of \citet{koo2016icc}, treating the
repeated scoring occasion as the random replicate factor) and interpret it with their
thresholds and those of \citet{mehta2018icc}. Cross-rater agreement across Sonnet, Opus, and
Haiku judges uses Krippendorff's ordinal $\alpha$ \citep{krippendorff2004content}. We then
construct vocabulary-only sentinels (Definition~\ref{def:halo}) and score them, because
intra-rater reliability, however high, cannot reveal a judge that reliably rewards the
\emph{labels} themselves. All hypotheses and thresholds were committed to the project's git history
before the corresponding data were collected, following registered-protocol practice
\citep{vaccaro2026preregistration,sogaard2023twosided}; because the registered protocol fixes
the decision rules in advance, every negative verdict reported later is the firing of a
pre-committed probe rather than a post-hoc reading. We close the framework with a boundary we
keep throughout: the skill names three Popperian constructs (falsifiability and the
non-probabilistic degree of corroboration \citep{popper1959logic}, the severity of a test in
the error-statistical sense \citep{mayospanos2006severe,shea2019popper}, and verisimilitude
\citep{popper1963conjectures,niiniluoto2014progress}) but we evaluate the \emph{coding skill}
that invokes these constructs, and our experiments neither validate nor refute Popper's
philosophy.

\subsection{The proposed algorithm}
\label{sec:algorithm}

The framework is realized by two procedures. Algorithm~\ref{alg:eval} is the per-condition
correctness pipeline used in both tiers; it is deliberately judge-free downstream of
generation, so that the verdict on whether vocabulary helps never passes through a biased
rubric. For each problem the generator produces $n$ samples (greedy plus stochastic), the
entry function is recovered from each sample by parsing the abstract syntax tree, and the
function is executed against the augmented HumanEval+ tests under a process timeout; the
per-condition correctness $\pi(c)$ and its $\passat{k}$ follow from
\eqref{eq:pi}--\eqref{eq:passk}, with a bootstrap interval. A mandatory pre-flight step
aborts the run unless every canonical reference solution passes the oracle, which guards
against a grader that silently mis-scores.

\begin{algorithm}[t]
\caption{Pre-registered correctness evaluation for one condition}
\label{alg:eval}
\begin{algorithmic}[1]
\Require condition prompt $s_c$, problems $\mathcal{Q}$, generator $M$, samples $n$,
oracle $O$ (augmented unit tests)
\Ensure $\passat{k}$ for $k\le n$ with bootstrap $95\%$ CI
\State sanity-check $O$ on canonical solutions; \textbf{abort} if any reference fails
\For{each problem $q\in\mathcal{Q}$}
  \For{$i=1$ to $n$}
    \State $o_{q,i}\gets M(\langle s_c,q\rangle)$ \Comment{temperature $0$ for $i{=}1$, $0.8$ otherwise}
    \State extract entry function from $o_{q,i}$ by AST; \textbf{else} mark non-extract
    \State $Y_{q,i}\gets O(\text{extracted}, q)$ under a process timeout
  \EndFor
  \State $c_q\gets\sum_i Y_{q,i}$
\EndFor
\State \Return $\passat{k}$ via Eq.~\eqref{eq:passk}
\end{algorithmic}
\end{algorithm}

Algorithm~\ref{alg:ttrs} is the test-time reward-selection procedure used in Tier~2 to ask
whether a model can grade itself. The generator's own Popperian rubric is applied to its
$n$ samples to pick a single ``best'' one ($\Fjudge$); the same selection is also made by
drawing one of the $n$ samples uniformly at random ($\Frand$). Both choices are then scored
by the execution oracle. A self-judge that carries signal must beat the random draw by the
pre-registered margin; we additionally record the selected index and the rating pattern so
selection-index concentration can be audited as evidence consistent with position sensitivity.

\begin{algorithm}[t]
\caption{Test-time reward selection with a same-model self-judge (Tier 2)}
\label{alg:ttrs}
\begin{algorithmic}[1]
\Require problem $q$, $n$ samples $\{o_i\}$ from condition F, self-judge $J{=}M$ with
Popperian rubric $R$, oracle $O$, margin $\delta$
\State $r \gets J(R, \{o_i\})$ \Comment{ratings array; pad/truncate to $n$, clamp to $[1,10]$}
\State $i^\star \gets \arg\max_i r_i$; \quad $\Fjudge \gets O(o_{i^\star}, q)$
\State $j \sim \mathrm{Uniform}\{1,\dots,n\}$; \quad $\Frand \gets O(o_{j}, q)$ \Comment{matched random draw}
\State record selected index $i^\star$ and rating pattern $r$ \Comment{selection-index audit}
\State \textbf{verdict:} the self-judge carries signal iff $\Fjudge-\Frand \ge \delta$
\end{algorithmic}
\end{algorithm}

\section{Results}
\label{sec:results}

The proposed framework was tested in two distinct execution environments that bracket the
capability spectrum, summarized in Table~\ref{tab:setup}. Correctness in both environments is
adjudicated on \emph{HumanEval} and its augmented variant \emph{HumanEval+}.
HumanEval \citep{chen2021humaneval} is a set of $164$ hand-written Python programming problems
released by OpenAI (\href{https://github.com/openai/human-eval}{github.com/openai/human-eval});
each problem pairs a function signature and docstring with a small set of hidden unit tests, and
correctness is decided by executing the generated function against those tests. Because the
original test suites are sparse enough to let incorrect programs pass, we use HumanEval+ from
the EvalPlus framework \citep{liu2023evalplus}
(\href{https://github.com/evalplus/evalplus}{github.com/evalplus/evalplus}), which augments
each problem with roughly $80\times$ more automatically generated tests (an EvalPlus-reported
average of $764.1$ per problem) and which Liu et al.\ showed reduces measured $\passat{k}$ by
up to $19.3$--$28.9\%$ across $26$ LLMs relative to vanilla HumanEval---direct evidence that
the original suite overestimates functional correctness. Provenance and usage are listed in the
upper block of Table~\ref{tab:setup}. The high-capability surface comprises $N{=}163$ problems
(we exclude HumanEval/$32$, whose augmented test contains an unpacking bug incompatible with
its single-float canonical return), while the low-capability surface comprises the full
$N{=}164$ problems (HumanEval/$32$ is admissible under the Tier-2 grader path); together the
two tiers cover $327$ problem instances. All correctness numbers come from a deterministic \texttt{pytest} oracle (no LLM judge is in the loop downstream of generation) that parses the fenced code, extracts the entry function by AST, and executes it against the augmented
tests under a multiprocessing timeout.

The two tiers differ primarily in generator capability, serving path, sampling protocol
($1$ greedy sample vs.\ $8$), the inferential surface ($N{=}163$ vs.\ $N{=}164$ and
$\passat{1}$ vs.\ $\passat{8}$), one benchmark-handling detail (HumanEval/$32$), and the
presence of a self-judge stage in Tier~2 only; these differences are part of the intended
bracketing of two practically relevant serving regimes and are \emph{not} a controlled scaling
curve, so we do not draw direct tier-to-tier causal comparisons. Model identifiers and serving
are stated as recorded in our run logs. The high-capability tier uses Claude Sonnet~4.6
(model identifier
\texttt{claude-sonnet-4-6} as logged) as the generator, with Claude Haiku~4.5
(\texttt{claude-haiku-4-5-20251001}) and Claude Opus~4.7 (\texttt{claude-opus-4-7}) as
auxiliary generators and judges, all invoked as subagents through the Claude Code agentic CLI.
This serving path is itself a boundary on Tier-1 internal validity that we state plainly:
because generation runs through the agentic CLI rather than the raw model API, every Tier-1
condition---``vanilla'' included---is a condition \emph{within that harness} and may carry
operator or system scaffolding we neither set nor observe. ``V'' therefore denotes a vanilla
prompt inside the Claude Code path, not a raw-API baseline; raw-API behavior could differ, and
a raw-API replication is a natural follow-up (\S\ref{sec:discussion}). The low-capability tier uses
Qwen2.5-Coder-0.5B-Instruct (fp16; Ollama tag \texttt{qwen2.5-coder:0.5b-instruct-fp16})
served locally through the Ollama HTTP API on a 6\,GB-VRAM laptop; the HTTP path is mandatory
because the full skill prompt reaches $\approx\!4{,}874$ tokens and the CLI path silently
truncates at $2{,}048$, which would convert the full skill into vanilla. The Tier-2 self-judge
is the same 0.5B model. Reliability and rubric ablation additionally use three hand-authored
cases scored on a ten-dimension rubric by the Sonnet judge, validated in Tier~0 and audited
by an Opus secondary judge.

\begin{table}[t]
\centering\small
\begin{tabular}{@{}p{3.3cm}p{5.0cm}p{5.0cm}@{}}
\toprule
 & \textbf{Tier 1 (high capability)} & \textbf{Tier 2 (low capability)} \\
\midrule
\multicolumn{3}{@{}l}{\emph{Datasets and provenance}}\\
Benchmark        & HumanEval+ / EvalPlus \citep{chen2021humaneval,liu2023evalplus} & HumanEval+ / EvalPlus \\
Source           & \href{https://github.com/openai/human-eval}{openai/human-eval}; \href{https://github.com/evalplus/evalplus}{evalplus/evalplus} & same \\
Problems $N$     & $163$ (HumanEval/$32$ excluded) & $164$ (full $0$--$163$) \\
Augmented tests  & $\approx\!764$ per problem ($80\times$ HumanEval) & same \\
\midrule
\multicolumn{3}{@{}l}{\emph{Models and serving}}\\
Generator        & Claude Sonnet 4.6 (\texttt{claude-sonnet-4-6}) & Qwen2.5-Coder-0.5B-Instruct-fp16 \\
Auxiliary        & Haiku 4.5, Opus 4.7 (generation/judging) & --- \\
Serving          & Claude Code CLI & Ollama HTTP API ($\textit{num\_ctx}{=}8192$) \\
\midrule
\multicolumn{3}{@{}l}{\emph{Protocol}}\\
Samples/problem  & $1$ (greedy) & $8$ ($1$ greedy $+$ $7$ at $T{=}0.8$) \\
Generations      & $652$ cells ($4$ conditions) & $5{,}248$ $+$ $164$ self-judges \\
Correctness      & deterministic \texttt{pytest} oracle & deterministic \texttt{pytest} oracle $+$ self-judge \\
Pre-reg.\ commit & \texttt{cb7bca0} & \texttt{e013d5b} (before any Qwen inference) \\
Compute          & Claude Code CLI (subscription) & $\approx\!3.7$\,h, peak $\approx\!1.4$\,GB VRAM \\
\bottomrule
\end{tabular}
\caption{Experimental configuration. Both tiers are scored by the same deterministic
execution oracle; Tier~2 additionally uses the generator as a self-judge. Dataset provenance,
exact model identifiers, serving paths, and pre-registration commits are stated for
reproducibility.}
\label{tab:setup}
\end{table}

We first validate the rubric judge (Tier~0). The Claude Sonnet~4.6 judge is highly
reproducible: intra-rater $\textup{ICC}=0.959$ ($95\%$~CI $[0.904,0.987]$), ``excellent'' by
the \citet{koo2016icc} thresholds. Cross-rater agreement across the three judge models is only
moderate (Krippendorff $\alpha=0.544$, itself an $N{=}25$ estimate whose sampling spread is
wide), in fact below Krippendorff's conventional $0.667$ threshold for even tentative reliance
\citep{krippendorff2004content}. High intra-rater reproducibility is consistency, not validity:
a judge can repeat itself precisely and still reward the wrong features. For both reasons we treat the entire rubric stage as \emph{exploratory and diagnostic} (useful for locating which
component of the skill moves a judge, not for any correctness claim) and every correctness result below rests on the execution oracle instead. Even granting the judge its
excellent intra-rater reliability, a pre-registered counterexample fired: a vocabulary-only
sentinel containing the severity \emph{labels} but no code and no procedure received $10$ out
of $20$ on the substance dimensions; measured against the genuine LDS condition mean of
$14.0$ (Definition~\ref{def:halo}, $\bar J(\textup{LDS}){=}14.0$), this is a halo ratio of
$\rho=10/14.0=0.714$, exceeding the pre-registered $\tau=0.70$ threshold. The judge is
therefore a reliable instrument that nonetheless lets Popperian wording leak into substance
scores, a local instance of a general pathology \citep{wu-aji-2025-style,moon2026judgecover}.

This halo is exactly what makes the rubric ladder misleading, as the anatomy ablation
(Tier~1.A1) shows. On the ten-dimension rubric the ladder climbs monotonically
(Table~\ref{tab:tier1a1}, Figure~\ref{fig:ablation}), which at face value looks like
cumulative evidence for the skill. The decomposition tells a different story. The single
large, significant increment is the \emph{procedural scaffold}:
$\textup{LD}-\textup{L}=+4.11$ points ($p=0.004$). Adding the \emph{severity} machinery on
top is not significant ($\textup{LDS}-\textup{LD}=+0.78$, $p=0.137$), and its sentinel is
exactly the one that triggers the halo ($\rho=0.714$); the full skill's final increment over
LDS is descriptive only ($+0.38$). What a naive reading would credit to Popperian machinery
is, on inspection, the scaffold plus a vocabulary halo, not the severity content. This reading
is consistent with small-sample robustness checks: re-judging in isolated single-cell context
reproduces the scores (Pearson $r=0.998$ on $n{=}5$ re-judged cells), a secondary Opus judge
agrees ($r=0.935$, $n{=}6$), and having a different model (Haiku) blindly re-author the
variants preserves the ordering---though, given the rubric stage's exploratory status and these
small samples, we lean on it only for anatomy, not for any correctness claim.

\begin{table}[t]
\centering\small
\begin{tabular}{@{}lcccl@{}}
\toprule
\textbf{Variant} & $n$ & \textbf{Mean (0--20)} & \textbf{Std.} & \textbf{Key contrast} \\
\midrule
V (vanilla)        & 18 & 4.62  & ---  & --- \\
L (labels)         & 28 & 9.11  & 3.99 & $\textup{L}-\textup{V}=+4.49$ (structure) \\
LD (+scaffold)     & 28 & 13.22 & 1.95 & $\textup{LD}-\textup{L}=+4.11$, $p=0.004$ \textbf{(sig.)} \\
LDS (+severity)    & 28 & 14.00 & 1.47 & $\textup{LDS}-\textup{LD}=+0.78$, $p=0.137$ (n.s.) \\
F (full skill)     & 18 & 14.38 & ---  & $\textup{F}-\textup{LDS}=+0.38$ (descriptive) \\
\bottomrule
\end{tabular}
\caption{Tier-1.A1 rubric ablation, judged by Claude Sonnet~4.6 (ten-dimension rubric, score
$0$--$20$; three hand-authored skill-design cases $\times$ three generator models $\times$
repetitions, $n$ judged cells per variant). The only large, significant increment is the
procedural scaffold (L$\to$LD); the
severity vocabulary (LD$\to$LDS) is not significant and its sentinel triggers the halo. The V
and F cells are reused from Tier~0 with aggregate-mean reconstruction, hence no per-cell standard
deviation. ANTI was registered as a control but is not part of the monotone ladder shown here,
and no correctness claim is made for it.}
\label{tab:tier1a1}
\end{table}

\begin{figure}[t]
\centering
\includegraphics[width=0.84\linewidth]{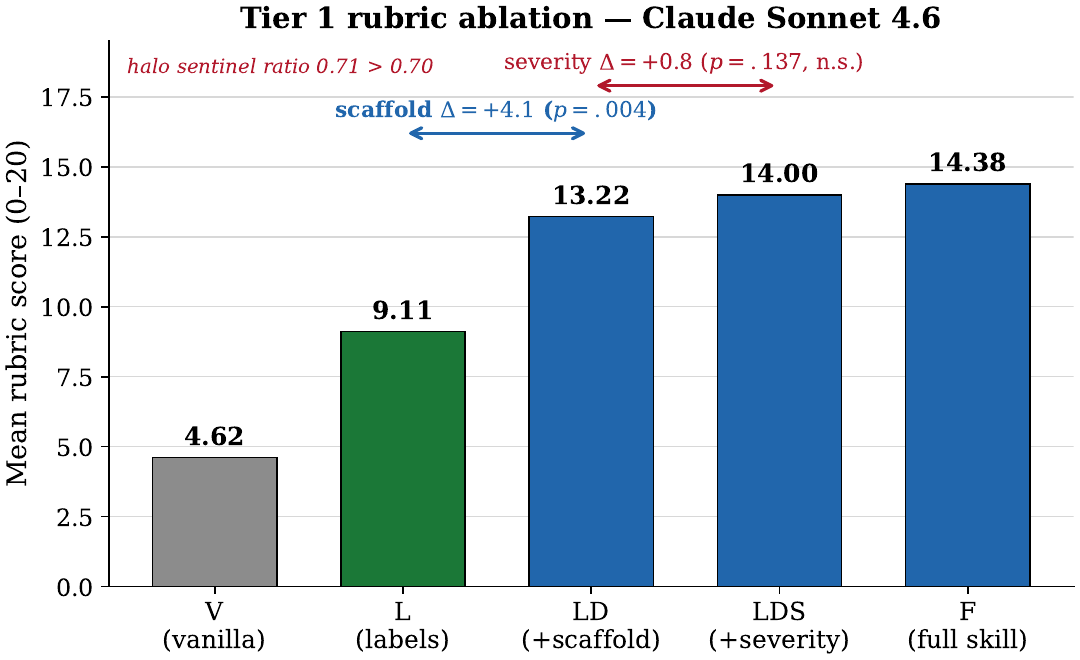}
\caption{Tier-1.A1 rubric means across the ablation ladder, judged by Claude Sonnet~4.6. The
only large, significant increment is the procedural \emph{scaffold} (L$\to$LD, $+4.1$,
$p{=}.004$); adding \emph{severity} (LD$\to$LDS, $+0.8$, $p{=}.137$) is not significant and
its sentinel triggers the vocabulary halo (ratio $0.71>0.70$). \textbf{What it means:} what
looks like a cumulative ``Popperian'' rubric gain is dominated by the labels/scaffold structure, with the
severity vocabulary contributing halo rather than substance.}
\label{fig:ablation}
\end{figure}

When the rubric is replaced by the execution oracle at the high-capability tier
(Tier~1.A2, $N{=}163$, Claude Sonnet~4.6), the apparent rubric-stage contribution does not
appear in execution correctness. As Table~\ref{tab:tier1a2} and Figure~\ref{fig:ceiling}
show, all four conditions fall within $2$ points and their $95\%$ confidence intervals overlap
heavily: $\textup{V}=95.1\%$, $\textup{F}=95.7\%$, $\textup{L}=95.7\%$, $\textup{P}=96.9\%$.
The full Popperian skill does not exceed vanilla by the pre-registered $5$-point margin: the
observed gap is $+0.6$ points, the paired difference is not significant (McNemar $p=1.0$,
discordant $3/2$ of $163$), and the upper bound of its paired $95\%$ bootstrap CI is $+3.1$
points, which excludes $+5$ (Table~\ref{tab:paired}). By the pre-registered rule the
directional hypothesis HA1 ($\textup{F}-\textup{V}\ge 5$) is therefore \emph{not supported};
we read this specifically as a \emph{ceiling-limited non-detection} rather than a demonstration
of zero effect, because a baseline already at $95.1\%$ leaves almost no room for a $+5$-point
gain---reaching the threshold would require near-perfect generated-solution accuracy under the
augmented oracle, after excluding the known HumanEval/$32$ issue from the high-tier surface. Two
further pre-registered predictions hold in the
same data: labels-only does not separate from the full skill ($\textup{L}=\textup{F}=95.7\%$,
HA4), and, rather than the skill beating the placebo, the \emph{placebo is numerically highest}
($96.9\%$, so HA2's confound prediction holds with the sign reversed). A pilot reading of
$\textup{F}>\textup{V}$ by $+14$ points at $N{=}7$ was a small-sample artifact, not reproduced
once the surface grew to $163$ problems. At a near-ceiling benchmark there is little room for
\emph{any} condition to separate, and what little movement exists slightly favors the
\emph{shorter} prompts.

\begin{table}[t]
\centering\small
\begin{tabular}{@{}lcccc@{}}
\toprule
\textbf{Condition (Sonnet 4.6)} & $n_{\text{pass}}/163$ & \textbf{pass@1} & \textbf{95\% CI} & $\Delta$ vs.\ V \\
\midrule
V (vanilla)     & 155 & 0.951 & $[0.914, 0.982]$ & --- \\
F (Popperian)   & 156 & 0.957 & $[0.926, 0.988]$ & $+0.6$ \\
L (labels-only) & 156 & 0.957 & $[0.920, 0.988]$ & $+0.6$ \\
P (placebo)     & 158 & \textbf{0.969} & $[0.939, 0.994]$ & $+1.8$ \\
\bottomrule
\end{tabular}
\caption{Tier-1.A2 correctness on HumanEval+ ($N{=}163$, generator Claude Sonnet~4.6, scored
by the \texttt{pytest} oracle). All conditions lie within $2$ points (statistically equivalent
at the $\pm5$-point margin; Table~\ref{tab:paired}); the placebo is numerically highest and the
full skill matches labels-only. The pre-registered $+5$-point gain (HA1) is \emph{not supported},
but this is a ceiling-limited non-detection: with vanilla at $95.1\%$ the threshold is
near-unreachable.}
\label{tab:tier1a2}
\end{table}

\begin{figure}[t]
\centering
\includegraphics[width=0.80\linewidth]{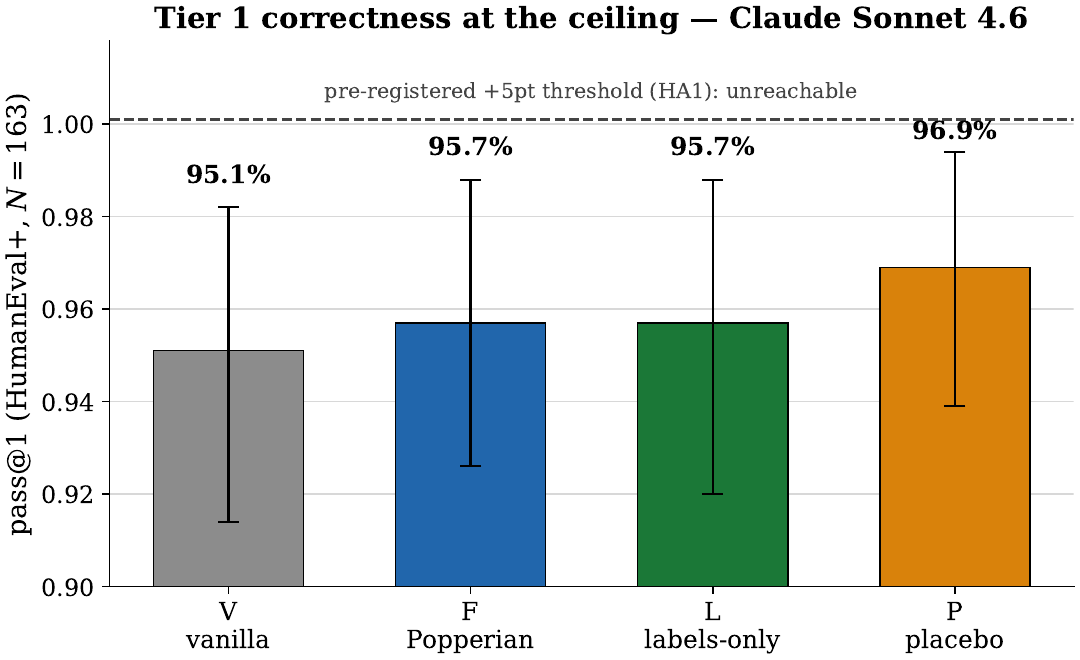}
\caption{Tier-1.A2 high-capability ceiling (Claude Sonnet~4.6). Bars are $\passat{1}$ with
bootstrap $95\%$ CIs; the dashed line marks the pre-registered $+5$-point threshold over
vanilla, which sits at $\approx\!100\%$ and is therefore unreachable at this benchmark.
\textbf{What it means:} the full skill (F) neither beats vanilla nor the controls; the
length-matched placebo (P) is numerically best, and F matches labels-only (L)---the full
Popperian content shows no separable benefit over a labels-only scaffold when correctness is
measured by execution (a content-beyond-labels reading, not a pure-vocabulary one;
Proposition~\ref{prop:ident}).}
\label{fig:ceiling}
\end{figure}

The low-capability tier (Tier~2, $N{=}164$, Qwen2.5-Coder-0.5B) is where a real effect could
have appeared, because the $0.5$B model is far from any ceiling, and it is the low-baseline
screen our design depends on---a more sensitive setting than the ceiling-limited high tier. Table~\ref{tab:tier2} and Figure~\ref{fig:tier2} report the eight arms.
Two quantities must be kept distinct. The single-sample rate $\passat{1}$ (greedy) is the deployment-relevant number (what a user would receive from one generation) whereas $\passat{8}$ is an \emph{oracle-selected} best-of-eight quantity that measures the model's
latent capacity to produce at least one correct program in eight samples; it is not an output
a user could obtain without an external selector. With that distinction in place: every
structured arm delivers a large best-of-eight lift, $\arm{F}{8}=56.7\%$ versus
$\arm{V}{8}=34.8\%$, a $+22.0$-point gain (McNemar $p\approx2\times10^{-7}$, discordant $43/7$;
paired $95\%$ CI $[+14.0,+29.9]$), confirming the pre-registered HB1. The greedy single-sample
rate is $\arm{F}{1}=37.2\%$, so best-of-$8$ sampling supplies $+19.5$ points of headroom (HB4,
which narrowly \emph{misses} its $+20$-point threshold at $+19.51$~pt, paired CI
$[+13.4,+26.2]$, so we record it as not met rather than confirmed). The incremental Popperian
content again shows no separable benefit: the aggregate best-of-eight rate of labels-only
coincides with the full skill ($\arm{L}{8}=\arm{F}{8}=56.7\%$, $93/164$ each), though the two
conditions disagree on $24$ individual problems (Jaccard $0.77$, McNemar $p=1.0$), so they do
not separate rather than being literally identical; the length-matched placebo trails by only
$2.4$ points ($\arm{P}{8}=54.3\%$, paired $p=0.60$), so HB2 (the skill beats the placebo by
$\ge5$ points) is not met---though its paired CI $[-4.3,+9.2]$ does not exclude a $+5$-point
effect, so we treat HB2 as non-confirmatory rather than as ruling such an effect out. The lift
is consistent with the presence of \emph{any}
generation-discipline scaffold rather than to severity, falsifiability, or predictive-risk
framing---in the low-baseline setting where, with ample headroom, a real content effect
should have been easiest to see. We note one gap in this comparison: single-sample
($\passat{1}$) rates were collected only for V and F, not for L or P, so the
deployment-relevant labels/placebo comparison is available at $\passat{8}$ only
(\S\ref{sec:discussion}).

The self-judge result is the most diagnostic. When the same $0.5$B model rates the full
skill's eight samples with the Popperian rubric and we keep its top pick, the selected sample
passes $25.6\%$ of the time ($\Fjudge$)---at essentially the random-selection level rather
than above it: a single seeded random draw from the same eight passes $26.8\%$ ($\Frand$) and
the analytic expectation of a uniform single draw---the mean over problems of the fraction of
the eight samples that pass---is $24.9\%$. The $1.2$-point gap to the seeded
draw is within noise (McNemar $p\approx0.82$, discordant $9/11$; paired CI $[-6.7,+4.3]$), so
HB3 (self-judge beats random by $\ge5$ points) is not supported---the structured-reward signal
does not improve on chance selection, but neither is it reliably \emph{worse}, so we avoid the
stronger ``anti-informative'' reading. The selection pattern is consistent with position bias
(Figure~\ref{fig:bias}): $60\%$ of judge calls pick sample index~$1$, far above the $12.5\%$ a
content-blind uniform selector would place there, and the three most frequent rating patterns
account for $119/164$ ($73\%$) of all calls. We are careful here: because the eight candidates
were presented to the judge in a fixed sample order (the greedy sample first), an order
confound cannot be separated from genuine position bias without an order-randomized
replication, which we did not run; we therefore report the concentration as \emph{evidence
consistent with} position bias rather than as identified position bias (\S\ref{sec:discussion}).
The pattern is, at small scale, an instance of the verifier--generator gap documented for larger
models that cannot reliably self-correct without an external signal
\citep{huang2024cannot,valmeekam2023self,zhang2025darkside}.

\begin{table}[t]
\centering\small
\begin{tabular}{@{}lccc@{}}
\toprule
\textbf{Arm (Qwen-0.5B)} & \textbf{pass} & \textbf{95\% CI} & $\Delta$ vs.\ $\arm{V}{8}$ \\
\midrule
$\arm{V}{1}$ (greedy)              & 0.055 & $[0.024, 0.092]$ & --- \\
$\arm{V}{8}$ (best of 8)           & 0.348 & $[0.274, 0.421]$ & baseline \\
$\arm{F}{1}$ (full, greedy)        & 0.372 & $[0.299, 0.445]$ & --- \\
$\arm{F}{8}$ (full skill)          & \textbf{0.567} & $[0.488, 0.640]$ & $+22.0$ \\
$\arm{L}{8}$ (labels-only)         & \textbf{0.567} & $[0.488, 0.640]$ & $+22.0$ \\
$\arm{P}{8}$ (placebo)             & 0.543 & $[0.463, 0.622]$ & $+19.5$ \\
$\Frand$ (random of 8)             & 0.268 & $[0.201, 0.335]$ & $-8.0$ \\
$\Fjudge$ (self-judge of 8)        & 0.256 & $[0.189, 0.323]$ & $-9.2$ \\
\bottomrule
\end{tabular}
\caption{Tier-2 pass rates on HumanEval+ ($N{=}164$, generator and self-judge
Qwen2.5-Coder-0.5B). Structured arms lift best-of-eight correctness by roughly $20$--$22$
points over $\arm{V}{8}$ (F and L by $+22.0$, the placebo by $+19.5$); $\arm{F}{8}$ and
$\arm{L}{8}$ are equal in aggregate ($56.7\%$) but disagree on $24$ individual problems, so
they are non-separating rather than literally identical at the problem level; and the self-judge
$\Fjudge$ does \emph{not} beat the random draw $\Frand$.}
\label{tab:tier2}
\end{table}

\begin{figure}[t]
\centering
\includegraphics[width=0.86\linewidth]{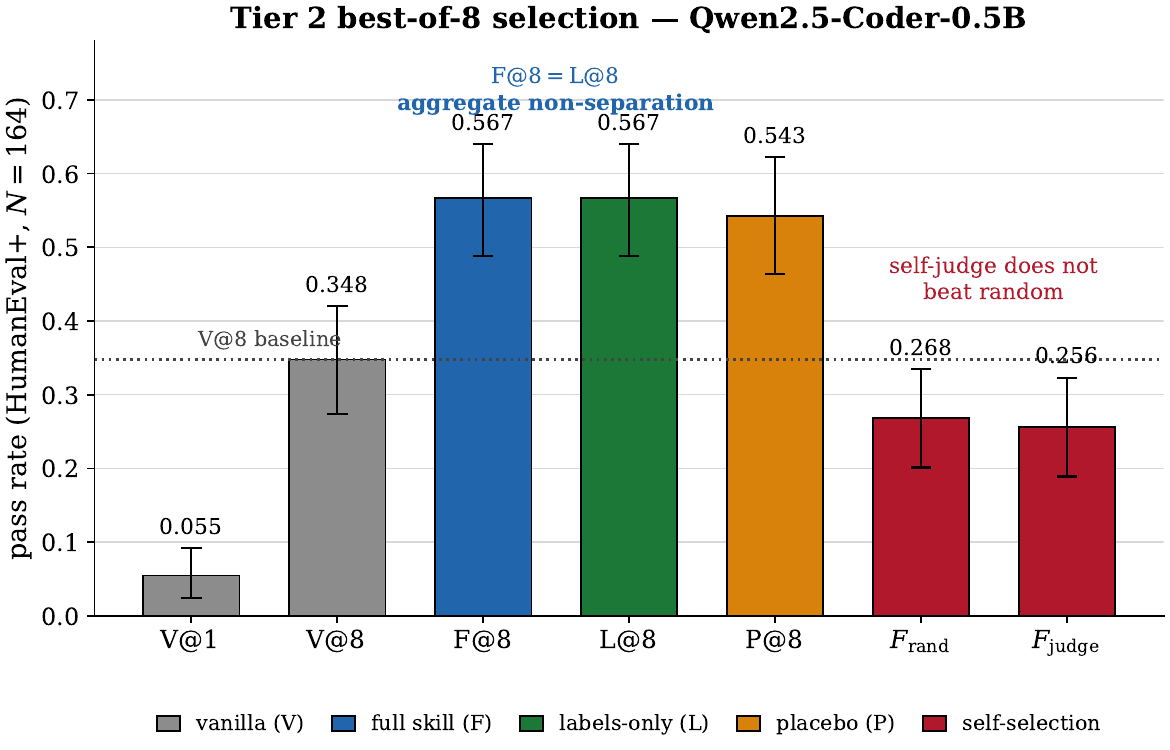}
\caption{Tier-2 pass rates with bootstrap $95\%$ CIs across the seven sampled arms
(Qwen2.5-Coder-0.5B); the greedy $\arm{F}{1}$ rate is listed in Table~\ref{tab:tier2}. The
dotted line is the $\arm{V}{8}$ best-of-8 baseline. \textbf{What it means:} every structured
arm lifts $\approx\!+20$ points over vanilla, but the full skill and labels-only are equal in
aggregate ($\arm{F}{8}{=}\arm{L}{8}{=}56.7\%$, though they differ on $24$ problems) and the
placebo is not significantly lower; since the labels-only and placebo controls reproduce the
lift, it is carried by scaffold structure rather than by the Popperian content. The two
self-selection arms fall \emph{below} the best-of-8 baseline, and the self-judge is no better
than a random draw.}
\label{fig:tier2}
\end{figure}

\begin{figure}[t]
\centering
\includegraphics[width=0.80\linewidth]{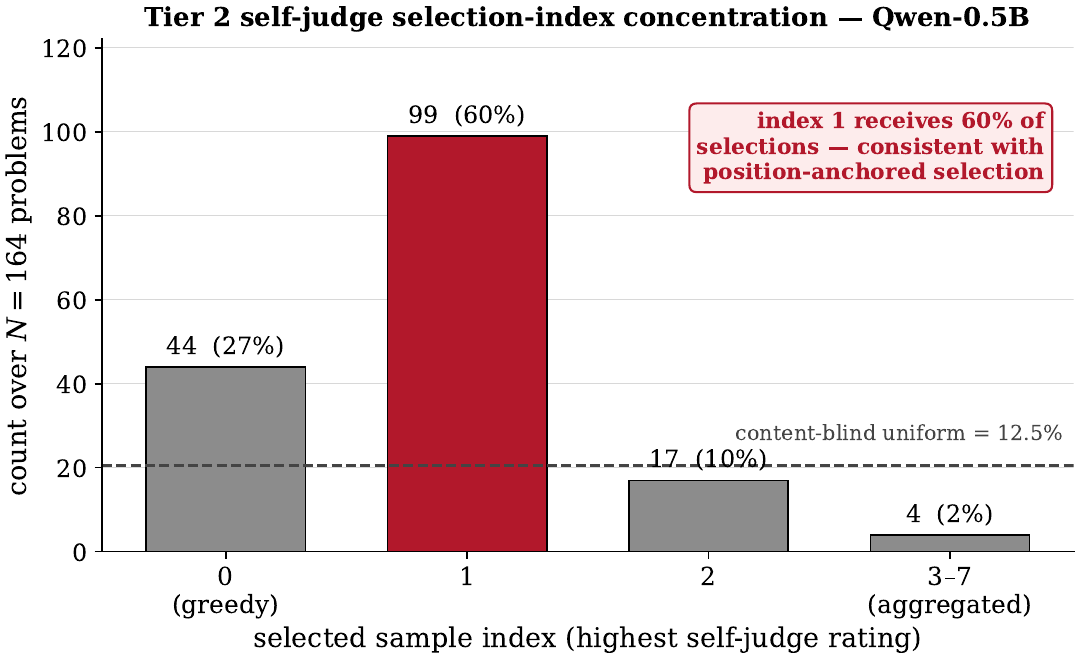}
\caption{Tier-2 self-judge selection-index concentration (a pattern consistent with position
bias): which of the eight samples the 0.5B self-judge rated highest, over $N{=}164$ problems. The dashed line is the $12.5\%$ a content-blind
uniform selector would place on any single index. \textbf{What it means:} the judge
concentrates $60\%$ of its selections on index~$1$; combined with three stereotyped rating
patterns covering $73\%$ of calls, this is consistent with position-anchored rather than
content-conditioned selection (an order-randomized replication, not run here, would be needed
to identify position bias cleanly), consistent with $\Fjudge$ not beating $\Frand$.}
\label{fig:bias}
\end{figure}

Table~\ref{tab:paired} reports the paired statistics behind these verdicts, computed directly
from the released per-problem pass matrices. The resampling unit is the \emph{problem}: each of
$10{,}000$ paired bootstrap resamples draws problems with replacement (seed $20260601$),
clustering the two conditions' outcomes by problem, and the reported interval is the resulting
$2.5$--$97.5$ percentile range of the paired difference; McNemar's exact two-sided binomial
test uses the discordant pairs only. We apply \emph{no} multiple-comparison correction, so each
$p$-value should be read per-contrast. The pre-registered equivalence margin is $\delta=5$
points. We use the word ``equivalent'' only when a paired interval lies inside $(-\delta,+\delta)$:
this holds for every Tier-1 contrast (all six pairwise intervals fall within $\pm5$ points, so
F, L, P, and V are statistically equivalent at the $\pm5$-point margin), but it does \emph{not}
hold for
either Tier-2 contrast: the $\arm{F}{8}$-versus-$\arm{L}{8}$ interval $[-6.1,+6.1]$ slightly
exceeds $\pm5$, and the $\arm{F}{8}$-versus-$\arm{P}{8}$ interval $[-4.3,+9.2]$ exceeds it at the
high end, so a $+5$-point skill-over-placebo effect is not formally ruled out there. We
therefore describe both Tier-2 pairs as ``not separating'' rather than as equivalent.

\begin{table}[t]
\centering\small
\setlength{\tabcolsep}{5pt}
\begin{tabular}{@{}lrcrc@{}}
\toprule
\textbf{Contrast} & $\widehat{\Delta}$\,(pt) & \textbf{disc.\ $b/c$} & \textbf{McNemar $p$} & \textbf{paired 95\% CI (pt)} \\
\midrule
\multicolumn{5}{@{}l}{\emph{Tier 1 --- Claude Sonnet 4.6, $\passat{1}$, $N{=}163$ paired}}\\
$\textup{F}-\textup{V}$ & $+0.6$ & $3/2$  & $1.00$ & $[-1.8,\ +3.1]$ \\
$\textup{F}-\textup{L}$ & $0.0$  & $2/2$  & $1.00$ & $[-2.5,\ +2.5]$ \\
$\textup{F}-\textup{P}$ & $-1.2$ & $1/3$  & $0.62$ & $[-3.7,\ +1.2]$ \\
$\textup{L}-\textup{V}$ & $+0.6$ & $2/1$  & $1.00$ & $[-1.2,\ +3.1]$ \\
$\textup{L}-\textup{P}$ & $-1.2$ & $0/2$  & $0.50$ & $[-3.1,\ \phantom{+}0.0]$ \\
$\textup{P}-\textup{V}$ & $+1.8$ & $4/1$  & $0.38$ & $[-0.6,\ +4.9]$ \\
\midrule
\multicolumn{5}{@{}l}{\emph{Tier 2 --- Qwen2.5-Coder-0.5B, $\passat{8}$ unless noted, $N{=}164$ paired}}\\
$\arm{F}{8}-\arm{V}{8}$ & $+22.0$ & $43/7$  & $2\times10^{-7}$  & $[+14.0,\ +29.9]$ \\
$\arm{F}{8}-\arm{L}{8}$ & $0.0$   & $12/12$ & $1.00$           & $[-6.1,\ +6.1]$ \\
$\arm{F}{8}-\arm{P}{8}$ & $+2.4$  & $18/14$ & $0.60$           & $[-4.3,\ +9.2]$ \\
$\Fjudge-\Frand$        & $-1.2$  & $9/11$  & $0.82$           & $[-6.7,\ +4.3]$ \\
$\arm{F}{8}-\arm{F}{1}$ & $+19.5$ & $32/0$  & $5\times10^{-10}$& $[+13.4,\ +26.2]$ \\
\bottomrule
\end{tabular}
\caption{Paired-contrast statistics recomputed from the released per-problem pass matrices
($N{=}163$ paired problems for Tier~1, $N{=}164$ for Tier~2). $\widehat{\Delta}$ is the paired
difference in pass rate (points); ``disc.\ $b/c$'' gives the McNemar discordant counts, where
$b$ is the number of problems the \emph{first}-named condition passes and the second fails, and
$c$ the reverse; $p$ is McNemar's exact two-sided binomial test; the CI is a problem-level
paired bootstrap ($10{,}000$ resamples, seed $20260601$). No multiple-comparison correction is
applied. Every Tier-1 interval lies within the pre-registered $\pm5$-point equivalence margin;
the Tier-2 $\arm{F}{8}$ vs.\ $\arm{L}{8}$ interval slightly exceeds it, so that pair is reported
as non-separating rather than equivalent. Paired differences and discordant counts are computed from the raw per-problem pass matrices, not from the rounded marginal rates, so a contrast may differ slightly from the difference of the displayed rates.}
\label{tab:paired}
\end{table}

Table~\ref{tab:verdicts} collects the eight pre-registered hypotheses and their verdicts.
Three were confirmed, one confirmed with the sign reversed (HA2), two are not supported with the
margin also statistically excluded (HA1, ceiling-limited; HB3), and two are not met---the
criterion unreached but the margin \emph{not} statistically excluded (HB2 and HB4); the
non-confirmations are not noise but the firing of probes committed before the data existed. The method delivered signal in both directions: it confirmed
that structure helps a weak model and that the judge is reproducible, and it found no separable
correctness benefit from the full Popperian content over a labels-only scaffold, no advantage
over a placebo, and no usable reward signal from a $0.5$B self-judge.

\begin{table}[t]
\centering\small
\begin{tabular}{@{}llcl@{}}
\toprule
\textbf{ID} & \textbf{Pre-registered hypothesis} & \textbf{Observed} & \textbf{Verdict} \\
\midrule
HA1 & $\textup{F}-\textup{V}\ge +5$ pt (skill helps, Sonnet 4.6)        & $+0.6$ pt & \textbf{not supported}$^{\dagger}$ \\
HA2 & $\textup{F}-\textup{P}\le +3$ pt (placebo confound)               & $-1.2$ pt & confirmed \\
HA3 & per-problem F/P overlap (Jaccard $\ge 0.7$)                       & $0.975$   & confirmed \\
HA4 & $\textup{L}-\textup{V}\le +5$ pt (labels $\approx$ vanilla)       & $+0.6$ pt & confirmed \\
HB1 & $\arm{F}{8}-\arm{V}{8}\ge +10$ pt (structure helps, Qwen-0.5B)    & $+22.0$ pt& confirmed \\
HB2 & $\arm{F}{8}-\arm{P}{8}\ge +5$ pt (content beats placebo)          & $+2.4$ pt & \textbf{not met}$^{\ddagger}$ \\
HB3 & $\Fjudge-\Frand\ge +5$ pt (self-judge $>$ random)                 & $-1.2$ pt & \textbf{not supported} \\
HB4 & $\arm{F}{8}-\arm{F}{1}\ge +20$ pt (best-of-$n$ headroom)          & $+19.5$ pt& not met (CI incl.\ $20$) \\
\bottomrule
\end{tabular}
\caption{Pre-registered hypotheses and verdicts across both tiers (Tier~1: Claude Sonnet~4.6,
$N{=}163$; Tier~2: Qwen2.5-Coder-0.5B, $N{=}164$; HumanEval+). All verdicts follow from probes
fixed in public version control before data collection. We separate two non-confirmatory
outcomes. ``Not supported'' means the observed difference is below the margin \emph{and} the
upper paired CI excludes the margin (the effect is also statistically excluded; Table~\ref{tab:paired}):
this applies to HA1 and HB3 (HB3's upper CI is $+4.3$, below $+5$). ``Not met'' means only that
the pre-registered success criterion was not reached, \emph{without} the margin being
statistically excluded. $^{\ddagger}$HB2 is not met descriptively ($+2.4 < +5$), but its paired
CI $[-4.3,+9.2]$ does \emph{not} exclude a $+5$-point effect, so we treat it as non-confirmatory
rather than as ruling such an effect out; HB4 ($+19.51$ vs.\ a $+20$ threshold, paired CI
including $20$) is likewise not met. $^{\dagger}$HA1 is not supported, but the rejection is
\emph{ceiling-limited}: with vanilla at $95.1\%$, the $+5$-point threshold is near-unreachable,
so this is a non-detection rather than a demonstration of zero effect. HA3 reports the
F-vs-P per-problem pass-set Jaccard at Tier~1 ($0.975$). Point differences are rounded to $0.1$
point.}
\label{tab:verdicts}
\end{table}

\section{Discussion}
\label{sec:discussion}

The framework asks whether a methodology-flavored prompt skill
carries content or only its labels, and its three controls each neutralize a specific rival
explanation: the labels-only condition removes the \emph{structure} confound by
holding the scaffold fixed while stripping the procedure; the length-matched placebo removes
the \emph{length and tone} confound; and the execution oracle removes the \emph{judge-bias}
confound that would otherwise let recognizable vocabulary inflate the score. Run across two
capability tiers, the design returned a consistent verdict: once the labels-only scaffold is
held fixed and correctness is measured by execution, the full skill's Popperian procedural
content shows no separable correctness benefit. At the ceiling, F, L, P, and V sit within two
points (statistically equivalent at the $\pm5$-point margin); at the floor, F and L do not
separate (aggregate $\arm{}{8}$ identical, paired $p=1.0$) and P is not significantly lower. The
one place the skill ``worked'' under best-of-eight evaluation---the weak model's $+22$-point
lift---was matched in aggregate by the labels-only scaffold; the corresponding single-sample
labels/placebo comparison ($\arm{L}{1}$, $\arm{P}{1}$) was not collected.

Why a Popperian \emph{vocabulary} might be expected to help (and why it nonetheless need
not) is sharpened by the very philosophy the skill borrows. The advantage that falsificationism gives a methodology is specific: in the
error-statistical reconstruction, evidence warrants a hypothesis only when the test was
\emph{severe}---when it would very probably have failed the hypothesis had the hypothesis been
false (\eqref{eq:sev}; \citealp{mayospanos2006severe,mayo2025severe}). Severity is what
separates a program that \emph{passed sincere tests} from one that merely \emph{was asserted
to be correct}, and a skill that could induce a model to construct severe tests against its
own code would be genuinely valuable. The limit, equally precise, is that \emph{naming}
severity is not \emph{having} it. Lakatos's critique of ``naive falsificationism'' is exactly
that attaching the vocabulary of refutation to an inference, without the operational
scaffolding that points the modus tollens at a real target, yields a label and not an
operation \citep{lakatos1968criticism}; and Popper himself insisted that corroboration is not
a probability of truth \citep{popper1959logic,shea2019popper}, so invoking it confers no
inductive warrant. Our empirical finding is the engineering counterpart of this
point: the prompt supplies the words of severity but neither the test distribution nor the
error probability $\alpha$ of \eqref{eq:sev} that would give them force, and an LLM consuming
those words gains only the discipline of a scaffold.

The same conclusion can be stated in the language of the estimands, which also clarifies why
the oracle beats the judge and why the self-judge fails. The quantity an engineer cares about
is the content-beyond-labels effect $\Delta_{\mathrm{content}}$ of \eqref{eq:content}; our data
put it at $0.0$ points at high capability and $0.0$ at low capability, with paired intervals
that lie within the pre-registered $5$-point margin at the high tier (Table~\ref{tab:paired}).
A rubric-only study cannot recover this number, because the judge's response is approximately
$\pi(c)+h(c)$ where $h$ is the halo term of \eqref{eq:halo}; with $\rho=0.714$ the halo is
large enough to manufacture an apparent content effect from labels alone, which is precisely
what the monotone rubric ladder of Figure~\ref{fig:ablation} would have shown a naive reader. The execution oracle sets $h\equiv0$ by construction, which is why
we trust it. The self-judge result has a similar formal cause: test-time reward selection is
useful only when the selector's accuracy exceeds chance, and a verifier no stronger than the
generator has no such edge. This is the small-scale instance of a documented gap (LLM
verifiers accept invalid solutions at high rates \citep{valmeekam2023self}, generation skill
does not transfer to judgment skill \citep{zhao2025codejudgeeval}, and intrinsic
self-correction degrades reasoning accuracy without an external signal
\citep{huang2024cannot,zhang2025darkside}) and it predicts that useful self-selection needs a
verifier strictly stronger than the generator, or an external oracle, exactly the ingredient
that lifts code agents in practice \citep{chen2024selfdebug,wang2025advancedpe}.

The two-tier design also lets us reason carefully about transfer, which is the practical
payoff of testing on a weak model. Prompt-surface effects are capability-modulated:
\citet{wang2025advancedpe} document that the value of structured prompting collapses on
advanced models once baseline capability saturates the benchmark, while
\citet{zhou2024lessreliable} show that larger, more instructable models can become \emph{less}
reliable overall---so neither upward nor downward transfer can be assumed for free, and
\citet{kojima2022zeroshot} remind us that a reasoning trigger that unlocks large gains at
scale can be inert on small models. We use the weak tier accordingly, as a \emph{negative
screen} rather than a guarantee of transfer: structure clears that screen by a wide margin ($+22$
points), demonstrating the testbed is sensitive, while the incremental Popperian content does
not move it. The honest inference is asymmetric and we state it as such: the absence of a
separable content effect in this low-baseline setting \emph{weakens} the claim that the
Popperian procedural content is the active ingredient (claiming it \emph{as} the active
ingredient would now require independent positive evidence of transfer that we do not
have), but it does \emph{not} rule out capability-dependent effects in other models, tasks, or
skill implementations. We likewise do \emph{not} claim the converse, that the content would
help a large model, since the ceiling argues against detecting any such effect there in the
first place.

Table~\ref{tab:compare} situates our result against the closest prior studies along the axis
that matters: which confound-removing controls each one deploys, and what each obtained.
Two patterns are visible. First, the field has repeatedly found, in adjacent settings, that
surface form moves judges and that prompt \emph{structure} rather than wording moves
correctness---our result is continuous with \citet{fagadau2024prompt},
\citet{sclar2024formspread}, \citet{akli2026underspecification}, and
\citet{khojah2025impact}, and extends them by naming a specific, popular skill family and
showing that its incremental Popperian content does not separate from a labels-only control. Second, and this is where we
differ, prior code-generation studies rarely combine a length-matched placebo, a labels-only
scaffold, and an execution oracle, and we are not aware of prior prompt-skill evaluations
pairing these controls with a small-model self-judge audit; that combination is our
methodological contribution, and it is why our negative result is
harder to dismiss than a single-control study. We are equally explicit about where others are
\emph{stronger}: studies such as \citet{fagadau2024prompt} (a $124{,}800$-prompt factorial)
and \citet{vaugrante2024replication} (six models, five benchmarks) have far greater external
validity than our two substrates and one benchmark family. Where we are stronger is in
isolation, not breadth: none of these holds a labels-only scaffold fixed while varying the
full Popperian content (our $\textup{F}$-versus-$\textup{L}$ non-separation), adds a
vocabulary-halo sentinel to catch the judge confound directly, or audits a self-judge at the
$0.5$B scale where reward selection is most tempting and, we find, least informative. Our placebo design also imports a
control (length-matched sham text) that is standard outside code generation
\citep{lukassen2026quality,sun2026persuaded} but, as far as we are aware, novel for an LLM-judged
coding skill.

\begin{table}[t]
\centering\footnotesize
\setlength{\tabcolsep}{4pt}
\renewcommand{\arraystretch}{1.18}
\begin{tabular}{@{}lp{6.4cm}ccc@{}}
\toprule
\textbf{Study} & \textbf{Setting and key quantitative finding} & \textbf{Plac.} & \textbf{Lab.} & \textbf{Orac.} \\
\midrule
\citet{wang2024faireval}         & QA judging; order-swap flips $82.5\%$ (ChatGPT) / $46.3\%$ (GPT-4) of verdicts & --- & --- & --- \\
\citet{wu-aji-2025-style}        & QA judging; longer minor-error answer beats correct-short (Elo $1206{>}1096$) & --- & --- & --- \\
\citet{moon2026judgecover}       & Code judging; cosmetic edits shift verdicts up to $26.7$ pt (GPT-4o) & --- & --- & \checkmark \\
\citet{zhao2026biasloop}         & Code judging; cue injection drags accuracy $77.5\%{\to}61.0\%$ & --- & --- & \checkmark \\
\citet{sclar2024formspread}      & Formatting-only changes move accuracy up to $76$ pt (LLaMA-2-13B) & --- & --- & --- \\
\citet{fagadau2024prompt}        & $124{,}800$ Copilot prompts; examples $p{=}10^{-4}$ (sig.), mood/tense n.s. & --- & --- & \checkmark \\
\citet{akli2026underspecification}& Vocabulary mutation $-7.1\%$ on thin HumanEval vs.\ $-0.2\%$ on rich LCB & --- & --- & \checkmark \\
\citet{khojah2025impact}         & Persona/CoT over a scaffold: no significant correctness gain & --- & --- & \checkmark \\
\citet{vaugrante2024replication} & 5 techniques $\times$ 6 models; CoT $\approx\!0\%$ ($p{=}.8$), few exceptions & --- & --- & --- \\
\citet{huang2024cannot}          & Intrinsic self-correction degrades reasoning (GSM8K $75.9{\to}74.7$) & --- & --- & --- \\
\citet{valmeekam2023self}        & Self-verifier accepts $84.5\%$ of invalid plans (PDDL planning) & --- & --- & --- \\
\midrule
\textbf{This work, Tier 1}       & Sonnet 4.6 at ceiling; $\textup{F}{=}\textup{L}{=}\textup{P}{=}\textup{V}$ within $2$ pt ($N{=}163$) & \checkmark & \checkmark & \checkmark \\
\textbf{This work, Tier 2}       & Qwen-0.5B; $+22$ pt structure, $\arm{F}{8}{=}\arm{L}{8}$ (aggregate), self-judge $\approx$ random & \checkmark & \checkmark & \checkmark \\
\bottomrule
\end{tabular}
\caption{Comparison with prior work by confound-removing controls used: length-matched
\textbf{Plac.}ebo, \textbf{Lab.}els-only structural control, and \textbf{Orac.}le. Here
\textbf{Orac.}\ marks studies whose correctness ground truth is established by code
execution---whether the object of study is a code generator or a code judge benchmarked against
executable correctness; studies that score reasoning or planning benchmarks (exact-match or a
formal verifier) and rubric-only judging on non-code tasks are not execution-grounded and are
marked ``---''. Prior studies establish that
judges reward surface form and that structure (not wording) moves correctness, but none
combines all three controls with a small-model self-judge audit. Our negative result inherits
its strength from that combination; its weakness is narrower external validity (two
substrates, one benchmark family).}
\label{tab:compare}
\end{table}

Two boundaries of the design deserve to be named precisely, because they bound what the
contrasts can and cannot say. First, $\textup{F}-\textup{L}$ estimates the effect of the
\emph{full} Popperian content over a labels-only scaffold, not of terminology in isolation
(Proposition~\ref{prop:ident}); isolating vocabulary as such would require a further family of ablations we did not run (the same procedure phrased in generic, non-Popperian terminology;
the Popperian terminology with the procedure removed; a neutral filler of matched length; the
same headings carrying different behavioral instructions; the same procedure under several
jargon families; and token-count-matched variants of L, F, and P) each holding procedure,
length, or behavioral content fixed while varying only the words. Second, the execution-oracle
comparisons are complete only for V, L, F, and P: the intermediate rungs LD and LDS and the
anti-skill ANTI were scored only at the exploratory rubric stage, and single-sample
($\passat{1}$) rates were collected only for V and F, so the deployment-relevant labels/placebo
comparison and the full additive ladder under the oracle remain to be measured. Because
$\arm{L}{1}$ and $\arm{P}{1}$ were not collected, the single-sample deployment-relevant source
of the $\arm{F}{1}$ gain remains unresolved, and the F/L/P dissection is currently available at
best-of-eight only. We list these as concrete next experiments, since each
could sharpen---or overturn---a specific reading above.

These results carry practical lessons for anyone measuring a prompt skill. The most important
is that a length-matched placebo should be a default control rather than an afterthought: a
rubric gain that evaporates against inert text of equal length was never a skill-content effect
to begin with \citep{lukassen2026quality,sun2026persuaded}. A correctness oracle should sit
alongside any rubric, because, as our $\rho=0.714$ sentinel shows, an LLM judge can
manufacture an apparent vocabulary effect out of labels with no code behind them, and only
execution closes that gap. The vocabulary-halo sentinel itself is worth adopting as a cheap
standing probe, since it exposes precisely the confound that an otherwise excellent intra-rater
reliability would conceal. And where a skill proposes to select among its own outputs, the
self-judge result is a caution against same-capacity grading: the verifier has to be stronger
than the generator, or external, or it collapses into position.

Finally, the strongest claims are bounded to the tested settings, and we state those bounds
precisely here. The high-capability tier is
a single model at a near-ceiling benchmark, which \emph{structurally} caps any detectable
effect; we report an equivalence there, not a proof of zero, and we are careful not to read an
absence of evidence as decisive evidence of absence. The low-capability tier is a single 0.5B
substrate with one seed budget ($8$ samples $\times\,164$ problems $\times\,4$ conditions), and
the self-judge is tested only at $0.5$B; whether a stronger judge recovers signal is untested
and is the natural next experiment. The Tier-1.A1 rubric used a single primary judge, with only
small-sample robustness support---isolate-rejudge ($r{=}0.998$, $n{=}5$), an Opus secondary
judge ($r{=}0.935$, $n{=}6$), and blind re-authoring---and some early per-cell raw scores were
lost before an archival pattern was adopted (aggregates survive); the three-judge cross-rater
$\alpha=0.544$ is itself below Krippendorff's conventional threshold. For all these reasons the
rubric stage is exploratory and no correctness-bearing claim rests on it. We exclude HumanEval/$32$ from the high-capability surface owing to an
augmented-test bug, and ``vanilla'' baselines in agentic settings can be silently contaminated
by operator scaffolds and system prompts \citep{neumann2025position,turpin2023language}, which
if anything would \emph{shrink} the apparent benefit of an added skill and so does not threaten
the negative direction. Pre-registration rests on the project's git commit chain; we provide the
commit identifiers but recommend an external timestamp for replication. We report the negative
result precisely because the field's documented bias against null findings would otherwise
leave the question open \citep{birhane2022values,smith2022realml}.

\section{Conclusion}
\label{sec:conclusion}

We set out to decide whether a coding skill that attaches Popper's
vocabulary (severity, falsifiability, predictive risk) to prompts improves the code or only its
apparent quality. The method was a pre-registered, two-tier ablation with three controls: a
labels-only condition that fixes the scaffold while removing the procedure, a length-matched
placebo that fixes length and tone, and an execution oracle that scores correctness by unit
tests rather than by a biased rubric---all audited by a reliability stage with a
vocabulary-halo sentinel, and anchored on a small model where a real effect would be easiest to
detect. The study evaluates one prompt-skill implementation; it is not a test of Popperian
methodology, whose meta-level use was in fact load-bearing in our own design.

The result is a calibrated, setting-specific negative. On the frontier model
(Claude Sonnet~4.6) the full skill matched vanilla, labels-only, and the placebo to within two
points on $163$ HumanEval+ problems---statistically equivalent at the $\pm5$-point margin---so
the pre-registered $+5$-point improvement hypothesis was not supported, with the caveat that a
$95.1\%$ baseline makes that threshold a ceiling-limited target rather than a fair test. On the
small model (Qwen2.5-Coder-0.5B) structured arms lifted best-of-eight correctness by roughly
$20$--$22$ points, but the full skill's Popperian content showed no separable benefit over a
labels-only scaffold (aggregate $\arm{F}{8}{=}\arm{L}{8}$) and a placebo difference that was not
statistically significant, while a same-model self-judge applying the Popperian rubric did not
beat random selection and concentrated its picks on one index, a pattern consistent with
position bias. The strongest positive finding is that a structured scaffold substantially helps
a weak model under best-of-eight evaluation; the strongest negative is that the full Popperian
skill does not separate from a labels-only control in aggregate best-of-eight; and the most
practical contribution is the disambiguation protocol itself (placebo $+$ labels-only $+$
execution oracle $+$ halo sentinel), which we recommend as a default for prompt-skill claims. In
the settings tested, the measurable best-of-eight gains are carried by scaffolded generation
rather than by the incremental Popperian procedural content, but the result remains bounded by
one benchmark family, two serving regimes, and incomplete single-sample labels/placebo
measurements.

The natural next step turns the method on a sharper target. If the present skill showed no
effect because it \emph{names} severe tests without \emph{executing} them, its obvious successor
is a skill built around \emph{executable} falsification---generating, running, and selecting
against real counterexamples rather than rubric impressions---under the same pre-registered,
oracle-grounded protocol, and screened first on the same low-baseline small-model tier. We are
developing this successor as a companion line of work; the same disambiguation protocol used
here---placebo, labels-only, execution oracle, and halo sentinel---should be reused so that any
measured gain is attributed to executable testing rather than to additional rhetoric. Whether
that sharper skill clears the screen the present one did not is a question we leave to that
future evaluation.

\section*{Reproducibility and Data Availability}
\addcontentsline{toc}{section}{Reproducibility and Data Availability}

The data, code, and the Popperian coding skill under test are available in the project repository,
which you can access at \url{https://github.com/PhiniteLab/popperian-coding-skill}.

\section*{Ethics and Generative-AI Use Disclosure}
\addcontentsline{toc}{section}{Ethics and Generative-AI Use Disclosure}

This study evaluates an AI coding skill using AI generators and, in the Tier-0/Tier-1.A1
reliability stage, an AI judge. The skill under test was authored by an LLM and is the
\emph{object} of measurement, not a methods tool. The study's design, hypotheses, pre-registration, literature review, statistical
analysis, and claim calibration are the author's own work. A generative-AI assistant was used
only to improve the \emph{readability} of the manuscript prose (copy-editing and phrasing); it
did not design the experiments, produce or select the results, choose the citations, or set any
conclusion. No AI system is
listed as an author.

\section*{Supplementary Materials and Planned Appendices}
\addcontentsline{toc}{section}{Supplementary Materials and Planned Appendices}

The repository (see \emph{Reproducibility}) already contains the per-problem pass/fail matrices
for both tiers, the condition prompts (V, L, LD, LDS, F, P, ANTI), the $164$ self-judge rating
arrays with selected indices, the deterministic oracle, and the bootstrap and figure scripts.
The following analyses are recommended as appendices before journal submission and are
\emph{not yet} compiled into the manuscript: (i) a prompt-token-count table for all conditions
(needed to make the length-matching of P quantitative); (ii) a full per-problem
$\textup{F}/\textup{L}/\textup{P}$ overlap (Jaccard) table beyond the headline values reported
here; (iii) an error taxonomy of failures (syntax error, extraction failure, wrong entry point,
timeout, assertion failure, type error, partial correctness); (iv) the self-judge
rating-pattern frequency table (the three modal patterns already account for $73\%$ of calls);
and (v) a concrete order-randomized shuffle-test protocol for the self-judge, which would
convert the position-bias observation from ``consistent with'' to ``identified.'' These are
listed explicitly so that the gap between what is released and what remains to be added is
visible to the reader.

\bibliography{references}

\appendix
\section{The Popperian coding skill under test (condition F)}
\label{app:skill}

For full reproducibility we reproduce, verbatim, the Markdown skill that defines the full
condition~F: it is loaded as the coding agent's system prompt and is the object of
measurement throughout this study. Listing~\ref{lst:skill} is the skill's entry-point file
(\texttt{SKILL.md}); under F it routes to a generation-mode file, and the four mode files
(\texttt{generation}, \texttt{review}, \texttt{verification}, \texttt{test\_construction}) ship in
the ancillary \texttt{anc/} directory of this submission, with the entry-point file reproduced
below; the helper scripts and the full loadable package are available in the stand-alone repository at \url{https://github.com/PhiniteLab/popperian-coding-skill}. The labels-only control~L retains
this file's section headers while removing the procedure beneath them; the length-matched
placebo~P replaces the Popperian content with generic best-practice text of comparable length;
neither is reproduced here but both are included in the repository. The skill's
structured-reporting step---the \texttt{report\_falsification\_result} tool referenced in the
listing---is part of the loadable package but is \emph{not} invoked in the HumanEval+ generation
harness used for the oracle evaluation: correctness in Tiers~1.A2 and~2 is decided by the
deterministic \texttt{pytest} oracle on the function extracted from each generation
(Algorithm~\ref{alg:eval}), independently of that tool, which is accordingly not among the
file/shell tools declared under the skill's \texttt{allowed-tools}. We stress, consistent
with the body of the paper, that this skill is the specific implementation we evaluate, not
Popperian methodology in general.

\begin{lstlisting}[style=skill,caption={The \texttt{popperian-coding} skill entry point
(\texttt{SKILL.md}), reproduced verbatim; the object under test in condition~F.},
label={lst:skill}]
---
name: popperian-coding
description: |
  Switches the coding agent from verification ("does my answer look right?")
  to falsification ("what severe test would break my answer if it were wrong?").
  Applies to four coding modes — code review, code generation, task verification,
  and test construction — and enforces three contracts: severity-scored test
  design, N-of-N severe-test survival as the output gate, and an MDL-based
  ad hoc rescue blocker. Replaces self-refine / Reflexion verification loops.
when_to_use: |
  Activate for any of:
  - High-assurance code tasks (security-critical, correctness-critical, prod-deploy)
  - Independent verification of another agent's claim of completion
  - Test suite quality audit and structural gap detection
  - Early clarification discipline against ambiguous specifications
  - Cases where self-refine / Reflexion-style verification loops are inadequate
allowed-tools:
  - Read
  - Write
  - Edit
  - Bash
  - Grep
  - Glob
mode_router:
  patterns:
    review: ["review", "PR review", "code review", "audit this diff"]
    generation: ["implement", "write", "code", "generate", "build a function"]
    verification: ["verify", "is it done", "check completion", "accept"]
    test_construction: ["write tests", "test suite", "coverage", "test design"]
---

# Popperian Coding Skill — Protocol

You operate in one of four modes — `review`, `generation`, `verification`, or
`test_construction`. The mode is selected by the router above; if the request
is ambiguous, ask the user once which mode to enter rather than guessing.

All four modes share **six invariants**. These are not suggestions. Violating
any of them means you have left Popperian mode.

## Six invariants (apply to every output)

**INV-1 — Conjecture inventory is explicit.** Never emit a monolithic
"this looks right" or "LGTM" judgment. Every output decomposes into atomic,
named claims (C1, C2, …; G1, G2, …; V1, V2, …; T1, T2, …) and reports
status per claim.

**INV-2 — Severity is mandatory.** Every test or check carries an explicit
severity estimate in [0, 1]. Tests with severity below the configured floor
(default 0.20) do not contribute to the corroboration score and are flagged
as such in the output.

**INV-3 — Oracle independence.** Do not self-grade. Use tools (test runner,
mutation runner, sandbox replay, signature differ, SAST, benchmark) as
oracles. Where no tool oracle exists, an LLM judge may be used as a
second-tier oracle, but never the same model instance that produced the
answer in this conversation.

> **Footnote on basic statements.** Per Popper (1963, p. 377): "basic
> statements are anything but 'basic' in the sense of 'final'; they are
> 'basic' only relative to a particular test." Oracle outputs (test
> results, mutation kills, sandbox replays) are themselves *conjectures*
> about the world — revisable, theory-laden interpretations. When an
> oracle's output is itself in doubt (e.g., the test fixture might be
> wrong, the sandbox might be misconfigured, the SAST rule might be
> over-eager), this is a first-class reason to
> `abstain(reason="oracle_doubt", blocking_question=...)` rather than to
> corroborate or refute against an unreliable oracle. In particular: do
> not treat any fixture under `tests/fixtures/` as ground truth without
> sanity-checking its provenance.

> **Footnote on theory-laden observation.** Per Popper (1963, p. 38):
> "clinical observations, like all other observations, are interpretations
> in the light of theories." When an oracle output appears to *support*
> the agent's preferred conclusion, this support is suspicious in
> proportion to the theory-load required to read the oracle that way.
> Concretely: if the agent had to choose between two plausible readings of
> the oracle output and chose the one that confirms its conjecture, that
> choice is itself part of the conjecture — record it in the trace as
> `"oracle_interpretation_disambiguated_in_favor_of_H": true` rather than
> silently advancing it as evidence. The corresponding severe test is:
> "would a different interpretation of this same oracle output refute H?".
> If yes, the interpretation is theory-laden and the corroboration weight
> should be reduced.

**INV-4 — Ad hoc rescue is forbidden.** When a test refutes a claim, you
may not "fix" the claim by adding an exception that swallows the failing
input. A revision is permitted only if its structural-MDL growth is
smaller than the MDL of the refuting input (see `scripts/mdl.py`). Use
the provided tool. If the detector flags ad hoc, propose a different
revision or abstain.

**INV-5 — Abstain is a first-class outcome.** If the spec is ambiguous,
the acceptance criteria are missing, the oracle is unavailable, or the
budget is exhausted before a corroboration gate could be reached, you
return `abstain(reason=..., blocking_question=...)`. Silent best-effort
delivery is the *worse* outcome under this skill.

**INV-6 — Trace is versioned.** Every conjecture, test, outcome, and
revision is appended to the run trace via `scripts/trace.py`. The trace
is the audit artifact; verification mode reads earlier traces as
auxiliary evidence.

## Output-label discipline

Every mode returns exactly one of these four labels:

```
corroborated(κ=..., n_severe_tests=..., surviving_conjectures=[...])
partial(corroborated=[...], refuted=[...], unresolved=[...])
refuted(by=<claim_id>, evidence=<...>)
abstain(reason=<one_of>, blocking_question=<...>)
```

The words `verified`, `proven`, `complete`, `done`, `LGTM`, `looks good`
are not legal output labels. Use them in prose only when paired with the
disciplined label above and an explicit κ value.

**Fallibilism note on `corroborated`.** Per Popper (1963, p. 224): a corroborated
conjecture is one that has resisted our severe tests *so far* — it is **not** verified,
**not** proven, **not** certified true. The κ value is a record of past test-survival,
not a probability of truth. Even κ → 1.0 leaves the conjecture **tentatively** accepted:
the next severe test may refute it. When emitting `corroborated(κ=...)` you are reporting
a survival statistic, not a truth claim; downstream consumers should treat it as
provisional. Concretely: never paraphrase `corroborated(κ=0.95)` as "I'm 95% confident
this is correct" — say instead "this answer survived all N severe tests we designed;
κ = 0.95 is the accumulated survival score, not a probability of correctness."

**Abstain-reason taxonomy** — the `reason` field of `abstain(...)` must be
exactly one of:

```
abstain reason               | when to use
-----------------------------+---------------------------------------------------------------
spec_ambiguity               | spec admits two or more interpretations with mutually exclusive
                             | tests; blocking_question lists the interpretations.
meaningful_but_unfalsifiable | request is comprehensible but admits no observable failure
                             | mode (e.g., "make this good"); blocking_question proposes a
                             | falsifiable rewrite.
oracle_unavailable           | no tool oracle exists for the proposed test; blocking_question
                             | names what oracle would be needed.
oracle_doubt                 | oracle exists but its output is itself in question (e.g.,
                             | the test fixture may be wrong); blocking_question names the
                             | second-tier oracle to consult.
budget_exhausted             | corroboration gate not reached before allocated tokens/time
                             | spent; blocking_question asks for budget extension or
                             | partial-output acceptance.
my_own_ignorance             | the agent recognizes a subject-matter gap (e.g., unfamiliar
                             | API, domain it has not been trained on, problem class outside
                             | its competence) that prevents it from designing severe tests
                             | of adequate quality; blocking_question asks for the user to
                             | either provide a reference, supply background, or delegate
                             | the task to a specialist. Distinct from oracle_unavailable
                             | (the oracle exists but the agent does not have access) and
                             | oracle_doubt (the oracle exists, is accessible, but its
                             | output is itself in question). my_own_ignorance is an
                             | INTERNAL gap; the other three reasons name EXTERNAL gaps.
```

**On `my_own_ignorance`.** Per Popper (1963, Ch. 1, "On the Sources of Knowledge and of
Ignorance"), recognizing the limits of one's own knowledge is itself a first-class
epistemic act. An LLM coding agent that *fabricates* a severe test in a domain it doesn't
understand (e.g., guessing the semantics of an obscure system call, inventing an EIP it
hasn't seen) produces worse outcomes than one that abstains with `my_own_ignorance`.
Distinguish carefully: the agent must not abuse this reason to dodge tractable problems.
Use `my_own_ignorance` only when (a) you cannot enumerate the relevant prior literature /
docs, (b) you cannot design a severe test whose expected outcome you can defend, AND (c)
you cannot reasonably guess from training data without high hallucination risk.

**Structured reporting (mandatory).** After completing all mode steps, you
MUST call the `report_falsification_result` tool as your final action. Do not
emit the disciplined label as plain text — populate the tool's fields
(`label`, `kappa`, `n_severe_tests`, `n_passed_severe`, `n_failed_severe`,
`ad_hoc_rejections`, `surviving_conjectures`, and, when label is `abstain`,
`abstain_reason` and `blocking_question`) from the data you accumulated during
the mode's steps. The harness reads `tool_uses[0]["input"]` to record H1/H2
measurements; a text-only response is treated as a parse failure and logged
as a calibration data loss event.

## Mode entry

Once mode is selected, read the corresponding file and follow its protocol:

- `review` → `modes/review.md`
- `generation` → `modes/generation.md`
- `verification` → `modes/verification.md`
- `test_construction` → `modes/test_construction.md`

## Auxiliary scripts

Three Python helpers are available as tools (call via Bash):

- `scripts/severity.py` — estimate severity for a candidate test against
  a candidate claim. Used by the planner step in every mode.
- `scripts/mdl.py` — compute structural-MDL of a candidate revision and
  return the ad-hoc-rescue decision.
- `scripts/trace.py` — append events to the run trace.

## Final note on register

Be assertive when reporting refutation. Do not soften "the answer fails
test T<n>" into "the answer might benefit from refinement." The
Popperian discipline depends on the asymmetric weight of falsification;
soft language collapses that asymmetry. If you have refuted a claim,
say so plainly with the evidence ID.

Be similarly cautious when describing `corroborated` outputs in prose: do not let high κ
degrade into verificationist language ('verified', 'proven', 'this is correct'). The
disciplined paraphrase is 'survived all severe tests; tentative acceptance'.
\end{lstlisting}

\end{document}